\documentclass[11pt]{article}
\usepackage{epsfig}
\usepackage{latexsym}
\usepackage{amsmath}
\usepackage{verbatim}
\usepackage{enumerate}
\usepackage{setspace}
\usepackage{comment}
\usepackage{anysize}
\usepackage{epic}
\usepackage{eepic}
\usepackage{xcolor}
\usepackage{booktabs}

\usepackage[all,color,frame,import]{xy}
\newcommand{\tab}{\hspace{.25in}}

\newtheorem{theorem}{Theorem}[section]
\newtheorem{lemma}{Lemma}[section]
\newtheorem{proposition}{Proposition}[section]

\newtheorem{observation}{Observation}[section]

\newcommand\QED{\ifhmode\allowbreak\else\nobreak\fi
\quad\nobreak$\Box$\medbreak}

\newcommand{\proofstart}{\par\noindent \emph{Proof:} }
\newcommand{\proofend}{\QED\par}
\newenvironment{proof}{\proofstart}{\proofend}

\usepackage[letterpaper,hmargin=1in,vmargin=1in]{geometry}

\long\gdef\boxit#1{\vspace{5mm}\begingroup\vbox{\hrule\hbox{\vrule\kern5pt
\vbox{\kern5pt#1\kern5pt}\kern0pt\vrule}\hrule}\endgroup}

\begin{document}

\title{Speedup in the Traveling Repairman Problem with Constrained Time Windows}
\author{Greg N. Frederickson\thanks{Dept. of Computer Sciences, Purdue University, West Lafayette,
    IN 47907. {\tt gnf@cs.purdue.edu}} \and Barry Wittman\thanks{Dept. of Computer Science, Elizabethtown College, Elizabethtown,
    PA 17022. {\tt wittmanb@etown.edu}}
}

\maketitle

%\keywords{  approximation algorithms; time windows;  traveling repairman; TSP   }

%\MSCcodes{Primary: 68W25, 68Q25; Secondary: 90C35, 90C27 }
% See the MSC2000 codes at http://www.ams.org/msc/

\begin{abstract}
A bicriteria approximation algorithm is presented for the unrooted traveling repairman problem, realizing increased profit in return for increased speedup of repairman motion.  The algorithm generalizes previous results from the case in which all time windows are the same length to the case in which their lengths can range between l and 2.  This analysis can extend to any range of time window lengths, following our earlier techniques \cite{Frederickson6}.
This relationship between repairman profit and speedup is applicable over a range of values that is dependent on the cost of putting the input in an especially desirable form, involving what are called ``trimmed windows.''  For time windows with lengths between 1 and 2, the range of values for speedup $s$ for which our analysis holds is $1 \leq s \leq 6$.  In this range, we establish an approximation ratio that is constant for any specific value of $s$.\\

\end{abstract}

\section{Introduction}

In this paper we present an approximation algorithm for a practical time-sensitive routing problem, the unrooted traveling repairman problem with time windows.  The input to this problem is a speed at which a repairman can travel and a list of {\em service requests}.  Each service request is located at a node in a weighted metric graph, whose edges give the travel distance between nodes.  Each service request also has a specific time window during which it is valid for service.  The goal of the problem is to plan a route called a {\em service run} that, starting at any service request at any time, visits as many service requests as possible during their respective time windows.

Because the problem is NP-hard, our only hope for an efficient approach seems to be an approximation algorithm.
In the real world, a repairman may have some flexibility in choosing speed.  As a consequence, our earlier approximation algorithms \cite{Frederickson5} and this paper are parameterized by speedup $s$, so that we can characterize how much closer to optimal the repairman can do if he or she travels a factor of $s$ faster than a hypothetical repairman traveling along an optimal route at the baseline speed.  This type of approximation based on resource augmentation is well known in the scheduling community as shown by Bansal et al.~\cite{Bansal2}, Kalyanasundaram and Pruhs \cite{Kalyanasundaram}, and Phillips et al.~\cite{Phillips}.

The algorithms in this paper build on our earlier work \cite{Frederickson3, Frederickson6}, in which we introduced the first polynomial-time algorithms that give constant approximations to the traveling repairman problem when all the time windows are the same length.  As a counterpoint to the repairman problem, we also introduced the speeding deliveryman problem in \cite{Frederickson3, Frederickson6}, with an alternative optimization paradigm, namely speedup.  The input to the speeding deliveryman problem is the same as the input to the traveling repairman problem, but the goal is to find the minimum speed necessary to visit {\em all} service requests during their time windows and thus collect all profit.  In \cite{Frederickson6} we also gave constant-factor polynomial-time approximation algorithms for both problems when the time windows have lengths in some fixed range.

In both the repairman and deliveryman problems, our algorithms \cite{Frederickson3, Frederickson6} rely on {\em trimming} windows so that the resulting time windows are pairwise either identical or non-overlapping .  We trim time windows by repeatedly making divisions in time after a fixed amount of time has passed, starting at a specified time.  We define a {\em period} to be the time interval that starts at a particular division and continues up to the next division.  When time windows are unit length, we choose a period length of .5 time units.  Because we define periods so that no window starts on a period boundary, each time window will completely overlap exactly one period and partially overlap its two neighboring periods.  Trimming then removes those parts of each window that fall outside of the completely overlapped period.

In this simpler case where time windows all have the same length, the penalty for trimming the repairman is a reduction by a factor of $1/3$ in the number of requests serviced, 
and the penalty for the deliveryman is an increase by a factor of 4 in the speed needed to service all requests.
In \cite{Frederickson5}, we showed that, for unit time windows, a spectrum of performance is possible between these two extremes.
For some speedup $s$ greater than 1 but less than 4, we showed how to achieve an increase in the number of serviced requests,
proportional in some sense to $s$.  The approximation is also a function of graph property $\gamma$, where $\gamma = 1$ for a tree and $\gamma$ is no more than $2 + \epsilon$ for a metric graph, for any constant $\epsilon > 0$.  A more complete explanation of $\gamma$ is given in Sect.~\ref{section:trimming}.

In this paper, we extend our algorithms and analysis to the more challenging case in which time windows have lengths in some fixed range, specifically between 1 and 2.
We present an algorithm that finds approximations parameterized by speedup $s$ and property $\gamma$.
To prove these approximation bounds, our analysis establishes and takes advantage of the \textbf{existence} of an ensemble of runs that move backward and forward along the path of an (unknown) optimal run, similar to our work in \cite{Frederickson5}.  These runs are analyzed based on several different starting points for trimming.  To handle windows of different lengths (i.e., between 1 and 2),
we orchestrate several complementary trimming schemes, run our approximation algorithm on each combination,
and choose the best result.

On the surface, the approach we use to orchestrate trimming schemes is similar to our approach in \cite{Frederickson6}, which extended our earlier approximation algorithms from \cite{Frederickson3} to achieve a constant approximation on time windows whose lengths were between 1 and 2.  However, the similarity of the algorithms belies the fundamental difference in the analysis, whose complexity increases by at least an order of magnitude in the process of uniting speedup with non-uniform time windows.  The key to our algorithm remains using a different period length for each trimming scheme, with each subsequent scheme using a progressively longer period length.  Intuitively, by selecting the most profitable run found in any scheme, the algorithm adapts to different distributions of window sizes.  If most of the windows are short, a scheme of trimming to shorter lengths will be effective.  If most of the windows are long, a scheme of trimming to longer lengths will be effective.  Because the output of each trimming scheme is a set of trimmed windows of equal length, our speedup algorithms from \cite{Frederickson5} can then be applied directly.  As with the case of no speedup, we bound the approximation guarantee of our algorithm by  accounting for a variety of distributions of windows, but the tool needed to bound each distribution is now a considerably richer set of hypothetical runs.

The major contributions of this paper are two additional techniques needed to extend the analysis for speedup on unit-time windows to windows with non-uniform length.  The first technique is a significantly more complex design of ensembles to achieve good coverage, using a greater variety of runs, some of which have longer repeating patterns.  Once we select an appropriate ensemble, we use a symbolic description of the coverage of the runs in the ensemble to demonstrate good coverage for all speedups in the range of speedups under consideration.  The second technique is an approach for designing and coordinating together the different bounds of approximation as a function of speedup for different window lengths.  Using averaging arguments, we will show that any convex combination of the approximation guarantees for each trimming scheme is a lower bound on the profit of the best run produced by our algorithm.  For each range of speedups in question, we determine the best choices of weightings for a convex combination of the approximation bounds we have found.  By using the best convex combination of bounds from each scheme, we guarantee a good bound of approximation.  The details of these techniques are given for the case when window size is between 1 and 2, but other ranges of window size can be accommodated in a similar way.

As a result, we can still produce polynomial-time approximation algorithms with constant-factor approximations for a given $s$ over a significant speedup range.
Our process of combining together different approximation bounds, as a function of the speedup $s$, gives a final result in Table~\ref{table:speedupw12rats} that is more involved than our results in \cite{Frederickson5}.  The ratio has more piecewise ranges and its inverse is primarily nonlinear, even though the inverse of the ratio in each range is fairly close to a linear function.  Note that approximation ratios are typically defined to be at least 1, and so these approximation ratios will give the reciprocal of the fraction of profit collected at a given speedup.  For ease of presentation, most of the analysis in this paper will instead be in terms of the fraction of profit collected.

\medskip

\begin{table}[!hbt]
\begin{center}
\begin{tabular}{cc}
\underline{Upper Bound on Approximation Ratio} & \underline{Speedup} \medskip \\
$\begin{array}{c}
219\gamma\: /(26s+26)\medskip \\
\gamma\: (28s^2 + 24s + 12)/(5s^3 + 6s^2) \medskip \\
\gamma\: (-4s^3 + 40s^2 -12s + 8)/(s^4 - 2s^3 + 11s^2) \medskip \\
\gamma\: (68s^3 - 172s^2 - 140s - 92)/(11s^4 -21s^3 - 50s^2) \medskip \\
\gamma\: (292s^3-1636s^2+2672s-1472)/(39s^4-183s^3+180s^2) \medskip \\
\gamma\: (12s^2+8s+16)/(s^3+6s^2) \medskip \\
\gamma\: (-s + 16)/(s+4) \medskip \\
\gamma\: (3s - 26) /(s - 14)
\end{array}$ &
$\begin{array}{crclc}
1& \leq &s&  \leq &2 \medskip \\
2& \leq &s&  \leq &{7 \over 3} \medskip \\
{7 \over 3}& \leq &s&  \leq &{17 \over 7} \medskip \\
{17 \over 7}& \leq &s&  \leq &{5 \over 2} \medskip \\
{5 \over 2}& \leq &s&  \leq &3 \medskip \\
3& \leq &s& \leq &4 \medskip \\
4& \leq &s& \leq &5 \medskip \\
5& \leq &s& \leq &6
\end{array}$
\end{tabular}\\
\caption{Approximation ratios for speedup $s$ when time window lengths are between 1 and 2.}
\label{table:speedupw12rats}
\end{center}
\end{table}

Our results are recent developments in time-sensitive routing problems, which have received a lot of attention from the algorithms community in the last decade.  As with our particular problem, these problems typically identify the locations to be visited and the cost of traveling between them as the nodes and edges, respectively, of a weighted graph.  For example, the orienteering problem considered by Arkin et al.~\cite{Arkin}, Bansal et al.~\cite{Bansal}, Blum et al.~\cite{Blum3}, Chekuri et al.~\cite{Chekuri2}, and Chen and Har-Peled \cite{Chen} seeks to find a path that visits as many nodes as possible before a global time deadline.  The deadline traveling salesman problem which was also considered by Bansal et al.~\cite{Bansal} generalizes this problem further by allowing each location to have its own deadline.  Our traveling repairman problem can be viewed as a further generalization from a deadline to a time window.

A great deal of work by Bansal et al.~\cite{Bansal}, Bar-Yehuda et al.~\cite{Bar-Yehuda}, Chekuri and Kumar \cite{Chekuri}, Karuno et al.~\cite{Karuno3}, Tsitsiklis \cite{Tsitsiklis}, and the authors \cite{Frederickson6} has been done on the traveling repairman problem, although much of the preceding literature, including that from Bansal et al.~\cite{Bansal} and  Bar-Yehuda et al.~\cite{Bar-Yehuda}, considers the rooted version of the problem, in which the repairman starts at a specific location at a specific time. 

For general time windows in the rooted problem, an $O(\log^2 n)$-approximation is given by Bansal et al.~\cite{Bansal}.  An $O(\log L)$-approximation is given by Chekuri and Korula \cite{Chekuri3}, for the case that all time window start and end times are integers, where $L$ is the length of the longest time window.  In contrast, a constant approximation is given by Chekuri and Kumar \cite{Chekuri},
but only when there are a constant number of different time windows.  Our earlier work \cite{Frederickson6} and work by Chekuri and Korula \cite{Chekuri3} give $O(\log D)$-approximations to the unrooted problem with general time windows, where $D$ is the ratio of the length of largest time window to the length of the smallest.  Polylogarithmic approximation algorithms for the directed traveling salesman problem with time windows have also been given by Chekuri et al.~\cite{Chekuri2} and Nagarajan and Ravi \cite{Nagarajan2}.

\section{Trimming Time Windows and the Associated Loss}
\label{section:trimming}

In our earlier work \cite{Frederickson6, Frederickson5}, we trimmed time windows of unit length by first making divisions in time every $.5$ time units, starting at time 0.  We generalize the process for time windows with different lengths by instead making divisions every $\alpha$ time units, where $\alpha = .5$, $.75$, or $1$ for the range of time window lengths $[1,2)$.
Define a {\em period} to be the time interval that starts at a particular division and continues up to the next division.  In the case of unit time windows, each time window will completely overlap exactly one period and partially overlap its two neighboring periods, because we allow no window to start on a period boundary.  In the case of longer time windows, there will be different patterns of overlapping.  If a window completely overlaps with only one period, trimming will remove those parts of each window that fall outside of the completely overlapped period.  If a window completely overlaps with more than one period, one trimming scheme will remove all those parts of each such window that fall outside of the first completely overlapped period.  Another separate trimming scheme will remove all those parts of each such window that fall outside of the second completely overlapped period.  For long period sizes, some time windows may not completely overlap {\em any} full periods and will vanish in the process of trimming.  In each case, because periods do not overlap and a time window is trimmed to at most one period, trimmed time windows will be pairwise either identical or non-overlapping.  Our repairman algorithm in \cite{Frederickson6} identifies a variety of good paths inside each separate period and then uses dynamic programming to select and paste these paths together into a variety of longer good paths and ultimately a good service run for the whole problem.

To describe our results for both trees and general metric graphs, we use the graph property $\gamma$, where $\gamma = 1$ for a tree, derived in our earlier paper \cite{Frederickson6}, and $\gamma \leq 2 + \epsilon$ for a metric graph, derived by Chekuri et al.~\cite{Chekuri2}, for any constant $\epsilon > 0$.  To describe the running time for these repairman algorithms we use $\Gamma(n)$, where $\Gamma(n)$ is $O(n^4)$ for a tree and $O(n^{O(1/\epsilon^2)})$ for a metric graph.  The value of $\gamma$ and running time of $\Gamma(n)$ are dependent on the approximation bounds for finding maximum profit paths within a specific period on a specific class of graph.  Although the available results only give $\gamma$ values and $\Gamma(n)$ running times for trees and metric graphs, other classes of graphs, such as outerplanar or Euclidean graphs, may have intermediate values of $\gamma$ and $\Gamma(n)$.  

In \cite{Frederickson6} we showed that, for unit time windows with no speedup, the reduction due to trimming still allows us to visit at least $1/(3\gamma)$ of the pre-trimming optimal and, with a speedup of 4, we can visit at least $1/\gamma$ of the pre-trimming optimal.  In \cite{Frederickson5} we filled in the gap between these two extremes with an $6\gamma/(s + 1)$-approximation for speedup in the range $1 \leq s \leq 2$ and a $4\gamma/s$-approximation for speedup in the range $2 \leq s \leq 4$.  We continue to demonstrate the flexibility of trimming in the realm of speedup by extending these results to time windows with different lengths.

\section{The Ensemble Approach for Analyzing Performance}
\label{section:ensemble}

Given an instance of the repairman problem on unit time windows, our previous work \cite{Frederickson5} presented algorithms for rational speedup $s = q/r$ in the range $1 \leq s \leq 4$ that take $O(\min\{r, m\}\Gamma(n))$ time, where $m$ is the number of distinct periods.  Since the approximation function is smooth and continuous, those algorithms work for any real speedup $s$, in the same range, in  $O(m\Gamma(n))$ time.  

The analysis of these algorithms uses a number of different service runs on trimmed time windows that are based on moving backward and forward along an optimal tour $R^*$.  We rely on averaging over a suitable ensemble of runs to establish that some run $R$ on trimmed time windows does well.  Because we will build on this technique and, indeed, use some of the same runs from our earlier work \cite{Frederickson5} in ensembles for variable length windows, we will review our notation for describing these runs.

We define unit speed to be some reference speed.  Traveling with $s = 1$ is traveling at unit speed.  Our results will hold whenever unit speed is no faster than the slowest speed at which an optimal service run is able to visit all locations during their time windows.  Intuitively, this restriction means that we are focusing on those cases when unit speed is low enough that speeding up our service runs will actually give some benefit.
Let $R^*$ be an optimal service run at unit speed originally starting at time 0.

In our analysis
we use the term {\em racing} to describe movement, forwards and backwards, along $R^*$ at a speedup of $s$ times unit speed.  Note that our analysis of run coverage is described on a period length of $.5$ even though, in Sect.~\ref{section:performance of speedupw12}, we will apply this analysis to our algorithm, which uses three different period lengths.

Define service run $A$ as follows.  Start run $A$ at time $t = 0$ at the location that $R^*$ has at time $t = -0.5$.  Then run $A$ follows a pattern of racing forward along $R^*$ for $1$ period, racing backward along $R^*$ for $1 - 1/s$ periods, and then racing forward along $R^*$ for $1/s$ periods.  Note that the pattern of movement for run $A$ repeats every 2 periods.

Considering the problem in which windows have length between 1 and 2, let $\lambda$ be an upper bound on the number of periods fully contained in a window. Define $A^R$, the ``\emph{reverse}'' of $A$, as follows.  When run on a set of requests whose windows each fully contain at most $\lambda$ periods, run $A^R$ starts at time $t = 0$ at the location that $R^*$ has at time $t = \lambda/2$.  Then run $A^R$ follows a repeating pattern of racing forward along $R^*$ for $1/s$ periods, racing backward along $R^*$ for $1 - 1/s$ periods, and then racing forward along $R^*$ for $1$ period.  Figure \ref{figure:s=2} shows examples of runs $A$ and $A^R$ with a speedup of 2 when $\lambda = 1$.

\begin{figure}[!hbt]
\centering
\begin{xy}
\xyimport(397, 155){\includegraphics[width=.75\textwidth]{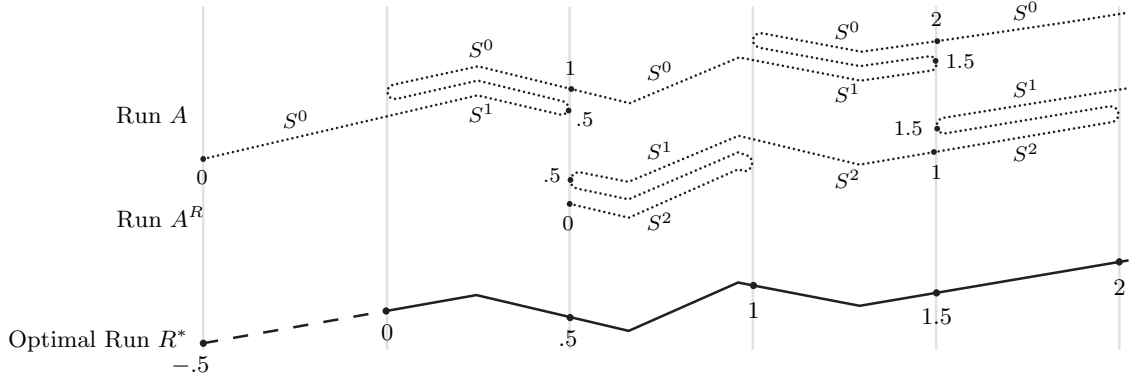}}
(2,60)*!R\txt\footnotesize{Run $A^R$};
(-5,4)*!R\txt\footnotesize{Optimal Run $R^*$};
(-5,106)*!R\txt\footnotesize{Run $A$};
%% Run A Periods
(41,104)*\txt\scriptsize{$S^0$};
(120,107)*\txt\scriptsize{$S^1$};
(120,136)*\txt\scriptsize{$S^0$};
(196,126)*\txt\scriptsize{$S^0$};
(276,116)*\txt\scriptsize{$S^1$};
(276,144)*\txt\scriptsize{$S^0$};
(352,153)*\txt\scriptsize{$S^0$};
%% Run B Periods
%(41,42)*\txt\scriptsize{$S^2$};
%(41,69)*\txt\scriptsize{$S^1$};
%(120,73)*\txt\scriptsize{$S^2$};
(196,58)*\txt\scriptsize{$S^2$};
(196,90)*\txt\scriptsize{$S^1$};
(276,77)*\txt\scriptsize{$S^2$};
(352,89)*\txt\scriptsize{$S^2$};
(352,118)*\txt\scriptsize{$S^1$};
%% Times
(80,8)*\txt\footnotesize{$0$};
(157,5)*\txt\footnotesize{$.5$};
(236,18)*\txt\footnotesize{$1$};
(314,15)*\txt\footnotesize{$1.5$};
%(-30,-6)*!L\txt\footnotesize{time $-.5$};
(-4,-7)*\txt\footnotesize{$-.5$};
(392,28)*\txt\footnotesize{$2$};
%% Run A \hat{t} times
%(-17,97)*!L\txt\footnotesize{time $0$};
(1,78)*\txt\scriptsize{$0$};
(164,104)*\txt\scriptsize{$.5$};
(158,127)*\txt\scriptsize{$1$};
(324,130)*\txt\scriptsize{$1.5$};
(314,149)*\txt\scriptsize{$2$};
%% Run B \hat{t} times
%(-3,55)*\txt\footnotesize{$-.5$};
%(-25,24)*!L\txt\footnotesize{time $-1$};
(157,57)*\txt\scriptsize{$0$};
(150,79)*\txt\scriptsize{$.5$};
(314,80)*\txt\scriptsize{$1$};
(302,100)*\txt\scriptsize{$1.5$};
\end{xy}
\caption{Examples of runs $A$ and $A^R$ with a speedup of 2 when $\lambda = 1$ based on an optimal run.  Times are labeled on the optimal run as well as runs $A$ and $A^R$.  Segments of each run are also designated $S^0$, $S^1$, and $S^2$ depending on which subset of runs they add coverage to.  This subset naming scheme will be fully explained in Sect.~\ref{section:lengths12}.}
\label{figure:s=2}
\end{figure}

For the purposes of analyzing our run $A$, number the periods $0$, $1$, $2$, and so on by the integer multiples of $.5$ that give their starting times.  Run $A$ repeats every 2 periods, and its coverage varies depending on whether the period number is even or odd.  To balance this asymmetry we define $\vec{A}$ and $\vec{A}^R$, shifted versions of $A$ and $A^R$, respectively.  Run $\vec{A}$ follows the same pattern as $A$ but starts the pattern at time $.5$ at the location $R^*$ has at time 0.  Run $\vec{A}^R$ follows the same pattern as $A^R$ but starts the pattern at time $.5$ at the location $R^*$ has at time $\lambda/2 + .5$.

Our analysis will require several versions of $A$ that have different starting points.  To simplify notation, for any given rational speedup $s = q/r$, let a {\em hop} be the amount of distance traveled in $1/(2r)$ time at unit speed.  Let $A_\Delta$ be the run $A$ moved forward $\Delta$ hops and let run $A_\Delta^R$ be the run $A^R$ moved backward $\Delta$ hops. Run $A_\Delta$ follows the same pattern of movement as run $A$ but starts at time $t = 0$ at the location that $R^*$ has at  time $t = -0.5 + \Delta/(2r)$.  Run $A_\Delta^R$, the reverse of $A_\Delta$, follows the same pattern of movement as $A^R$ but starts at $t = 0$ at the location that $R^*$ has at $t = \lambda/2 - \Delta/(2r)$.  These reversed and shifted versions of $A$ were required to establish the performance of our algorithms in \cite{Frederickson5} on unit length windows.  Although run $A$ with its 2-period repeating pattern is sufficient for those cases, we will introduce additional runs which repeat after 3 or 4 periods in order to handle windows of longer length.

If a service request $p$ is serviced by a run $R$ during the period that the time window of $p$ has been trimmed into, we say that $R$ {\em covers} $p$.
Let $S$ be a subset of service requests.  Define the {\em coverage} of $S$ by a run $R$, written $cover_R(S)$, to be the number of requests in $S$ covered by $R$ divided by the number of requests in $S$.  Define the coverage of $S$ by a set $U$ of runs, written $cover_U(S)$, to be the average of $cover_R(S)$ for every run $R \in U$.

We will still rely on our Average Coverage Proposition from our earlier work \cite{Frederickson5}:

\begin{proposition}[Average Coverage] \label{proposition:averaging}  Let $\{S_1, S_2, S_3, ... S_a\}$ be a collection of sets of service requests such that $\bigcup_i S_i$ gives all the service requests serviced by $R^*$ on untrimmed windows.  Let $U$ be a set of service runs.  If $\min_i\{cover_U(S_i)\} = \mu$, then there is at least one service run $\hat{R} \in U$ such that profit$(\hat{R}) \geq \mu \cdot$profit$( R^* )$.
\end{proposition}

The Average Coverage Proposition formalizes the following intuition.  Let a group of service runs achieve some coverage over a set of requests.  Let us also say that we have divided those requests into many different subsets, some of which may overlap, but the union of all the subsets is the original set of requests.  If we take the subset of requests with the worst average coverage, some service run in the group covers a fraction of total requests no smaller than that worst average coverage.  Otherwise, the average coverage of all subsets would be worse than the coverage of the worst covered subset, which is a contradiction.

Given a way of dividing requests into subsets, we wish to prove that some set of service runs achieves some lower bound on average coverage.  The following section will describe the algorithm we will use to find a service run and the analytical techniques we will use to establish a lower bound on its performance.  This analysis will depend on carefully showing an average coverage for various subsets of requests defined with respect to periods induced by trimming.

\section{Algorithm for Windows with Lengths between 1 and 2}
\label{section:lengths12}

In \cite{Frederickson6}, we describe algorithms for $s = 1$ that achieve constant approximations when window sizes are not necessarily uniform
but are close to being uniform.
We extend that approach for our speedup problem,
but specifically for windows whose lengths differ by at most a factor of 2.  From \cite{Frederickson5}, we give an algorithm called SPEEDUP that, for unit windows, finds a run of approximately optimal profit at speedup $s$.
Our approach is to modify SPEEDUP  and run it with three different sizes of period $\alpha$: .5, .75, and 1.

For each period size, we will consider multiple starting points for a set of periods, each spaced .25 apart.  
We modify SPEEDUP appropriately so that no window starts at the beginning of a period, for periods of size .5, .75, or 1.  This modified algorithm is called SPEEDUPW12 and is specified below.
In this algorithm, sets of periods whose $\alpha$ is .5, .75, or 1 can have 2, 3, or 4 unique starting positions, respectively.  
Depending on a given period size $\alpha$ and starting point, a window will partially fill 2 periods 
and fully fill 0, 1, 2, or 3 periods between the 2 partial periods.  For $\ell = 1, 2, 3$ and a specified value of $\alpha$, let $W_\ell$ be the set of windows
that completely fill exactly $\ell$ subintervals and partially overlap with two more of them.

\begin{table}[!htb]
\begin{tabular}{l}
\toprule
\textbf{SPEEDUPW12}\\
\midrule
\tab {\it PHASE 1:}\\
\tab Set $\alpha$ to .5 and identify windows for sets $W_1$, $W_2$, and $W_3$.\\
\tab For $i$ from $0$ to $1$,\\
\tab \tab Set the starting point for the periods to $i/4$.\\
\tab \tab For $j$ from $1$ to 2,\\
\tab \tab \tab For $k$ from $1$ to 3,\\
\tab \tab \tab \tab Trim each window in $W_1$ to its $1^{st}$ subinterval.\\
\tab \tab \tab \tab Trim each window in $W_2$ to its $j^{th}$ subinterval.\\
\tab \tab \tab \tab Trim each window in $W_3$ to its $k^{th}$ subinterval.\\
\tab \tab \tab \tab Run SPEEDUP and retain the best result so far.\\
\tab {\it PHASE 2:}\\
\tab Reset $\alpha$ to .75 and then identify windows for $W_1$ and $W_2$.\\
\tab For $i$ from $0$ to $2$,\\
\tab \tab Set the starting point for the periods to $i/4$.\\
\tab \tab For $j$ from $1$ to 2,\\
\tab \tab \tab Trim each window in $W_1$ to its $1^{st}$ subinterval.\\
\tab \tab \tab Trim each window in $W_2$ to its $j^{th}$ subinterval.\\
\tab \tab \tab Run SPEEDUP and retain the best result so far.\\
\tab {\it PHASE 3:}\\
\tab Reset $\alpha$ to 1 and then identify windows for $W_1$.\\
\tab For $i$ from $0$ to $3$,\\
\tab \tab Set the starting point for the periods to $i/4$.\\
\tab \tab Trim each window in $W_1$ to its $1^{st}$ subinterval.\\
\tab \tab Run SPEEDUP and retain the best result so far.\\
\bottomrule
\end{tabular}
\end{table}

When trimming, we choose from several choices of which single full subinterval to keep for each window.  For example, for periods of length .5 and for windows in $W_3$ which would have three full subintervals, the choices for trimming will be to trim the window down to either the first, second, or third full subinterval.  Combining these choices with the two choices associated with windows in $W_2$ and the single choice in windows in $W_1$ yields 6 trimmings.

The performance for speedup for windows in $W_1$ in the range $1 \leq s \leq 4$ is the same as the unit time window results given by our work in \cite{Frederickson5}.  In Sect.~\ref{section:performance W_2}, we give the performance for speedup for windows in $W_2$ in the range $1 \leq s \leq 5$.  In Sect.~\ref{section:performance W_3}, we give the performance for speedup for windows in $W_3$ in the range $1 \leq s \leq 6$.  Note that the SPEEDUP subroutine works for any real number $1 \leq s \leq 6$; however, our analysis will assume that $s$ is a rational number such that $s = q/r$.  In the case that $s$ is irrational, our analysis holds in the limit because the functions we find that bound performance in terms of $s$ are piecewise smooth and continuous.  It is worth repeating that our analysis uses only a period size of .5 but can still bound the performance of our algorithm with its three different period lengths by using careful accounting of subset coverage.

When dealing with windows of unit length in a previous paper \cite{Frederickson5}, we defined a partition of requests into three sets, based on which period a request was trimmed into versus which period an optimal run $R^*$ serviced the request in. Set $T$ consists of requests serviced by $R^*$ in the same period, set $E$ consists of requests serviced by $R^*$ in the preceding period, and set $L$ consists of requests serviced by $R^*$ in the following period.  For windows of length between 1 and 2, we need to extend this approach.  We will use a superscripted $S$ to designate that a request that was serviced by $R^*$ in either the first, second, third, fourth, or fifth periods with which a request overlaps.
For requests in $W_1$, the sets $L$, $T$, and $E$ will be renamed $S^0$, $S^1$, and $S^2$, respectively.  For requests in $W_2$ and $W_3$, we will go further and use designations $S^0$ through $S^3$ and $S^0$ through $S^4$, respectively. Set $S^3$ consists of requests serviced by $R^*$ two periods before the period into which those requests were trimmed, and $S^4$ consists of those requests serviced three periods earlier.
In the same earlier work \cite{Frederickson5}, we further partitioned $L$, $T$, and $E$ into $L_j$, $T_j$, and $E_j$ for $j = 1, 2, \ldots, r$.
In a similar way, we will partition sets $S^0$, $S^1$, $S^2$, $S^3$, and $S^4$ into $r$ equal-length divisions, subsets $S^0_j$, $S^1_j$, $S^2_j$, $S^3_j$, and $S^4_j$, for any given $j$, $j = 1, 2, \ldots, r$.

Let $[w, w+k)$ be any time window, where $1 \leq k \leq 2$.  Let $\omega$ be the smallest integer multiple of $1/(2r)$ that is greater than $w$.  We designate subintervals $[w,\omega)$, $[\omega, \omega + 1/(2r))$, $[\omega + 1/(2r), \omega + 2/(2r))$, $\ldots$, $[\omega + (4r - 1)/(2r), w + 2)$ by $w_0$, $w_1$, $w_2$, $\ldots$, $w_{4r}$.  For windows of length between 1 and 2 with a given choice of period starting times, all windows fall into set $W_1$, $W_2$, or $W_3$.  In our analysis, there is always an implied factor of $\gamma$ that accounts for the difference between the approximation on a tree and on a metric graph.

We now define a procedure called CREATE-TABLE-$\lambda$ that describes the process of determining coverage for a particular speedup $s$ for a particular run moved forward $\Delta$ hops.  This procedure is a generalization of our CREATE-TABLE procedure from \cite{Frederickson5} to $\lambda \geq 1$.  (Recall that $\lambda$ is an upper bound on the number of periods fully contained in a window.) Note that CREATE-TABLE-$\lambda$ is not an algorithm that is run in the process of finding an approximation to a repairman problem with speedup.  Rather, it provides a template that can be used to produce the tables used in analyzing the performance of such approximations.  So that the treatment here is self-contained, we repeat much of our discussion of table construction from \cite{Frederickson5}, modifying it as necessary so that it can also handle the additional types of runs that we will introduce.

Before CREATE-TABLE-$\lambda$ can be completely defined, it is necessary to explain the pattern of coverage generated by a run.  For the kind of runs we have seen so far, type $A$ runs, this pattern takes one of two forms.  Let $s$ be a rational number such that $s = q/r$.  Type $A$ runs repeat every two periods and thus can be represented with a pattern of coverage that uses a 1 to signify a subset covered every period and a 1/2 to signify a subset covered every other period.

Observe the movement of type $A$ runs, noting that, during its first period, such a run moves forward the same distance that an optimal run moves during $q$ subintervals.  During its second period of time, it moves backward the same distance than an optimal run moves during $q - r$ subintervals and then forward the same distance than an optimal run moves during $r$ subintervals. Then, the pattern repeats.  When $s < 2$, run $A$, during the first period in its pattern, covers $q$ successive subsets as it moves forward, while in its second period covers $r$ subsets as it moves forward.  Note that those subsets covered as $A$ moves backwards add nothing additional to the coverage.  Thus, this pattern of coverage is represented as $r$ repetitions of 1 and $q - r$ repetitions of 1/2.  When $s \geq 2$, run $A$, during the first period in its pattern, covers $q$ subsets subintervals, while in its second period covers $q - r$ subsets backwards but no new subsets forward.  This pattern of coverage is represented as $q - r$ repetitions of 1 and $r$ repetitions of 1/2. 

\begin{table}[!hbt]
\begin{tabular}{l}
\smallskip \\
\toprule
\textbf{CREATE-TABLE-$\lambda$(hops $\Delta$)} \\
\midrule
\tab Let the first element of the coverage pattern be indexed at 0.\\
\tab Number the subsets 0 through $r(\lambda + 1)$.\\
\tab Define function $\mathcal{C}$ based on the coverage pattern, such that:\\
\tab \tab $\mathcal{C}(i) = \left\{ \begin{array}{ll}
\sigma&\mbox{if term }i\mbox{ of the pattern is of value }\sigma\\ 
0&\mbox{otherwise}
\end{array} \right.$ \\
%\tab \tab $C(i) = 1$, ,\\
%\tab \tab $C(i) = 1/2$, if term $i$ of the pattern is 1/2, and\\
%\tab \tab $C(i) = 0$, for all other values of $i$\\
\tab Define function ${\mathcal F}$ on integers $i$, where $0 \leq i \leq r(\lambda + 1)$:\\
\tab \tab ${\mathcal F}(i) = \sum_{j = 0}^{r - 1} \mathcal{C}(i + j - \Delta)$.\\
\tab Define function ${\mathcal F}^R$ on the same domain:\\
\tab \tab ${\mathcal F}^R(i) = {\mathcal F}(r(\lambda + 1) - i)$.\\
\tab The final coverage function defined by the table is given by ${\mathcal F}(i) + {\mathcal F}^R(i)$.\\
\bottomrule
\end{tabular}
\end{table}

The values that $\mathcal{C}(i)$ can take on are dependent on the types of runs used.  For type $A$ runs, $\mathcal{C}(i)$ can be 0, 1/2, or 1.  Runs introduced later will have a larger range of values, but it is always the case that $0 \leq \mathcal{C}(i) \leq 1$.  Note that the functions $\mathcal{F}$ and $\mathcal{F}^R$ given in CREATE-TABLE-$\lambda$ are piecewise linear functions with ranges dependent on the fundamental pattern of coverage.  Due to its construction, the combination ${\mathcal F}(i) + {\mathcal F}^R(i)$ is also a piecewise linear function and symmetrical.  Thus, only the range $0 \leq i \leq \lfloor r(\lambda + 1)/2 \rfloor$ need be listed in tables.\\

Although CREATE-TABLE-$\lambda$ gives a procedure for creating a table for a given speedup, we need tables expressed symbolically to prove coverage for a range of speedups.  Instead of using specific numbers, we can leave the basic patterns of subset coverage for a given style of run (such as type $A$ runs) with its shifted version in terms of $q$ and $r$.  By shifting this pattern $r$ times and summing the results together, we account for the different alignments a time window might have with respect to the various subintervals.  This sum is the function $\mathcal{F}$, which can be expressed as a piecewise linear function.  Function $\mathcal{F}^R$, which describes reversed runs, can be similarly described.  To combine the two functions symbolically, we sort the end points of the subset ranges from both descriptions together.  If, for the given range of speedups being considered, there are two end points which cannot be ordered, we subdivide the range of speeds so that, in each new speed range, the two end points in question can be ordered.  Once the end points of each subset range have been sorted, combining the descriptions from the normal and reversed functions of the runs is achieved by simply summing each range.  We give an example of this process in Sect.~\ref{subsection:1<s<2 W_2}.

\section{Speedup Performance for Windows in Set $W_2$}
\label{section:performance W_2}

Recall that $W_2$ is the set of windows that completely fill exactly two periods.
We will now explore the speedup-performance trade-off for windows in $W_2$ for all speedups $1 \leq s \leq 5$.  For set $W_2$, our analysis must consider subsets $w_0$ through $w_{3r}$.  Throughout our analysis, we will we assign a 1 for full coverage and a $1/2$ for half coverage of any subset.  When examining the subsets for a given range of speedup values, the values are symmetrical around $w_{3r/2}$ when $r$ is even and symmetrical after $w_{(3r - 1)/2}$ when $r$ is odd.  Thus, the tables and proofs we use will not list contributions for subset $w_i$ where $i > \lfloor 3r/2 \rfloor$,  since the contribution at $w_i$ in these higher ranges is the same as the corresponding contribution at $w_{3r - i}$, by symmetry.

\subsection{Speedup $1 \leq s \leq 2$ for Windows in Set $W_2$}
\label{subsection:1<s<2 W_2}

For the range $1 \leq s \leq 2$, we can represent any rational speedup $s$ in the form $s = (r + k)/r$ with integers $r \geq 1$ and $0 \leq k \leq r$.  For this analysis, we consider service runs $A$, $A^R$, $A_{r - k}$, $A_{r - k}^R$, $A_{2r - k}$, and $A_{2r - k}^R$, noting that $\lambda = 2$.  Similar to $W_1$ for $1 \leq s \leq 2$, run $A$ covers set $S^0$ well, run $A^R$ covers set $S^3$ well, and the remaining four runs plug the holes left in the spotty coverage of sets $S^1$ and $S^2$.

We will pass over the simpler case for $A$ runs and use $A_{r - k}$ runs to give an example of how we construct symbolic coverage tables.  For speedup $s$ where $1 \leq s \leq 2$ and $\lambda = 2$, type $A$ runs have a fundamental pattern of coverage of $r$ subsets covered every period followed by $q - r = k$ subsets covered every other period.  Adjusting for the offset of $\Delta = r - k$ and making $r$ shifted, this pattern yields the values for $\mathcal{F}(i)$ and $\mathcal{F}^R(i)$ given in Table \ref{table:F functions for A_r-k runs}.

\begin{table}[!hbt]
\begin{center}
\begin{tabular}{rll}
$\mathcal{F}(i) = $ & $\left\{ \begin{array}{l}
k + i \smallskip \\ 
{3 \over 2}r - {1\over 2}k - {1\over 2}i \smallskip \\ 
2r - {1\over 2}k - i \smallskip \\
r - {1\over 2}i
\end{array} \right.$ &
$\begin{array}{rcl}
0 \leq &i& \leq r - k \smallskip \\ 
r - k \leq &i& \leq r \smallskip \\
r \leq &i& \leq 2r - k \smallskip \\
2r - k \leq &i& \leq 2r
\end{array}$
\\ \\
$\mathcal{F}^R(i) = $ & $\left\{ \begin{array}{l}
{1\over 2}i - {1\over 2}r \smallskip \\ 
i - r - {1\over 2}k \smallskip \\ 
{1\over 2}i - {1\over 2}k \smallskip \\
3r + k - i
\end{array} \right.$ &
$\begin{array}{rcl}
r \leq &i& \leq r + k \smallskip \\ 
r + k \leq &i& \leq 2r \smallskip \\
2r \leq &i& \leq 2r + k \smallskip \\
2r + k \leq &i& \leq 3r
\end{array}$
\end{tabular}
\end{center}
\caption{Separate coverage functions for $A_{r-k}$ and $A_{r-k}^R$ in $W_2$ when $1 \leq s \leq 2$.}
\label{table:F functions for A_r-k runs} 
\end{table}

Because $\mathcal{F}(i) + \mathcal{F}^R(i)$ is symmetric about $i = 3r/2$ if $r$ is even and after $i = (3r - 1)/2$ if $r$ is odd, we are only interested in the range $0 \leq i \leq \lfloor 3r/2 \rfloor$.  In this range, the sub-ranges $r \leq i \leq 2r - k$ and $2r - k \leq i \leq \lfloor 3r/2 \rfloor$ for $\mathcal{F}$ overlap with the sub-ranges $r \leq i \leq r + k$ and $r + k \leq i \leq \lfloor 3r/2 \rfloor$ for $\mathcal{F}^R$.  When $k \leq r - k$, then $r - k \leq r + k \leq \lfloor 3r/2 \rfloor$.  In that case, for $r + k \leq i \leq \lfloor 3r/2 \rfloor$, $\mathcal{F}(i) + \mathcal{F}^R(i) = (2r - k/2 - i) + (i - r - k/2) = r - k$, as in the last interval of the middle set of contributions in Table \ref{table:lowcontributions W_2}. When $k \geq r - k$, then $r - k \leq 2r - k \leq \lfloor 3r/2 \rfloor$.  In that case, for $2r - k \leq i \leq \lfloor 3r/2 \rfloor$, $\mathcal{F}(i) + \mathcal{F}^R(i) = (r - i/2) + (i/2 - r/2) = r/2$, as in the last interval of the middle set of contributions in Table  \ref{table:highcontributions W_2}.

Similar analysis for runs $A$ and $A^R$ and runs $A_{2r - k}$ and $A_{2r - k}^R$ produce the rest of Tables \ref{table:lowcontributions W_2} and \ref{table:highcontributions W_2}.  The combined coverages of runs $A$, $A^R$, $A_{r-k}$, $A_{r-k}^R$, $A_{2r - k}$, $A_{2r - k}^R$, and all of their respective shifted versions are all given in Table \ref{table:lowcontributions W_2} when $k \leq r - k$ and in Table \ref{table:highcontributions W_2} when $k \geq r - k$.

\begin{table}[!hbt]
\begin{center}
\begin{tabular}{rcll}
\begin{tabular}{r}
Combined contributions\\
for $A$ and $A^R$ 
\end{tabular}
& = & 
$\left\{ \begin{array}{l}
r - {1 \over 2}i \smallskip \\ 
r + {1 \over 2}k - i \smallskip \\ 
{1 \over 2}r + {1 \over 2}k  - {1 \over 2}i \smallskip \\
0
\end{array} \right.$ &
$\begin{array}{rcl}
0 \leq &i& \leq k \smallskip \\ 
k \leq &i& \leq r \smallskip \\
r \leq &i& \leq r + k \smallskip \\
r + k \leq &i& \leq \left\lfloor {3r \over 2} \right\rfloor
\end{array}$
\\ \\
\begin{tabular}{r}
Combined contributions\\
for $A_{r - k}$ and $A_{r - k}^R$
\end{tabular}
& = &
$\left\{ \begin{array}{l}
k + i \smallskip \\
{3 \over 2}r - {1\over 2}k - {1 \over 2}i \smallskip \\
r - k
\end{array} \right.$ &
$\begin{array}{rcl}
0 \leq &i&  \leq r - k \smallskip \\
r - k \leq &i&  \leq r + k \smallskip \\
r + k \leq &i& \leq \left\lfloor {3r \over 2} \right\rfloor
\end{array}$
\\ \\
\begin{tabular}{r}
Combined contributions\\
for $A_{2r - k}$ and $A_{2r - k}^R$ 
\end{tabular}
& = &
$\left\{ \begin{array}{l}
{1 \over 2}i \smallskip \\ 
i - {1 \over 2}k \smallskip \\ 
2i - r + {1 \over 2}k \smallskip \\
{3 \over 2}i - {1\over 2}r + {1\over 2}k \smallskip \\
r + 2k
\end{array} \right.$ &
$\begin{array}{rcl}
0 \leq &i& \leq k \smallskip \\ 
k \leq &i& \leq r - k \smallskip \\
r - k \leq &i& \leq r \smallskip \\
r \leq &i& \leq r + k \smallskip \\
r + k \leq &i& \leq \left\lfloor {3r \over 2} \right\rfloor
\end{array}$
\end{tabular}\\
\caption{Contributions of runs for windows in $W_2$ when $1 \leq s \leq 2$ and $k \leq r - k$.}
\label{table:lowcontributions W_2}
\end{center}
\end{table}

\begin{table}[!hbt]
\begin{center}
\begin{tabular}{rcll}
\begin{tabular}{r}
Combined contributions\\
for $A$ and $A^R$
\end{tabular}
& = & 
$\left\{ \begin{array}{l}
r - {1 \over 2}i \smallskip \\
r + {1 \over 2}k - i \smallskip \\
{1 \over 2}r + {1\over 2}k - {1 \over 2}i \smallskip \\
k - {1 \over 2}r
\end{array} \right.$ &
$\begin{array}{rcl}
0 \leq &i& \leq k \smallskip \\
k \leq &i& \leq r \smallskip \\
r \leq &i& \leq 2r - k \smallskip \\
2r - k \leq &i& \leq  \left\lfloor {3r \over 2} \right\rfloor
\end{array}$
\\ \\
\begin{tabular}{r}
Combined contributions\\
for $A_{r - k}$ and $A_{r - k}^R$
\end{tabular}
& = &
$\left\{ \begin{array}{l}
k + i \smallskip \\
{3 \over 2}r - {1 \over 2}k - {1 \over 2}i \smallskip \\
{1 \over 2}r
\end{array} \right.$ &
$\begin{array}{rcl}
0 \leq &i&  \leq r - k \\
r - k \leq &i& \leq 2r - k\\
2r - k \leq &i& \leq  \left\lfloor {3r \over 2} \right\rfloor
\end{array}$
\\ \\
\begin{tabular}{r}
Combined contributions\\
for $A_{2r - k}$ and $A_{2r - k}^R$ 
\end{tabular}
& = & 
$\left\{ \begin{array}{l}
{1 \over 2}i \smallskip \\ 
{3 \over 2}i - r + k \smallskip \\ 
2i - r + {1 \over 2}k \smallskip \\
{3 \over 2}i - {1\over 2}r + {1\over 2}k \smallskip \\
{5 \over 2}r - k
\end{array} \right.$ &
$\begin{array}{rcl}
0 \leq &i& \leq r - k \smallskip \\ 
r - k \leq &i& \leq k \smallskip \\
k \leq &i& \leq r \smallskip \\
r \leq &i& \leq 2r - k \smallskip \\
2r - k \leq &i& \leq  \left\lfloor {3r \over 2} \right\rfloor
\end{array}$
\end{tabular}\\
\caption{Contributions of runs for windows in $W_2$ when $1 \leq s \leq 2$ and $k \geq r - k$.}
\label{table:highcontributions W_2}
\end{center}
\end{table}

\begin{lemma}\label{lemma:lowestinterval W_2}
If the contributions from $A$ and $A^R$ are weighted by a factor of 2 and the contributions from $A_{r - k}$, $A_{r - k}^R$, $A_{2r - k}$, and $A_{2r - k}^R$ are weighted by a factor of 1, the yield for all intervals is at least $2r + k$.
\end{lemma}

\begin{proof}
We first consider the case when $k \leq r - k$, consulting Table \ref{table:lowcontributions W_2}.

If $0 \leq i \leq k$, then the yield for $w_i$ is $2r + k + i/2$, which is at least $2r + k$, since $i \geq 0$.

If $k \leq i \leq r - k$, then the yield for $w_i$ is $2r + 3k/2$, which is greater than $2r + k$.

If $r - k \leq i \leq r$, then the yield for $w_i$ is $5r/2 + k - i/2$, which is at least $2r + k$, since $i \leq r$.

If $r \leq i \leq \lfloor 3r/2 \rfloor$, then the yield for $w_i$ is $2r + k$.

\noindent We now consider the case when $k \geq r - k$, consulting Table \ref{table:highcontributions W_2}.  

If $0 \leq i \leq r - k$, then the yield for $w_i$ is $2r + k + i/2$, which is at least $2r + k$, since $i \geq 0$.

If $r - k \leq i \leq k$, then the yield for $w_i$ is $5r/2 + k/2$, which is at least $2r + k$, since $r \geq k$.

If $k \leq i \leq r$, then the yield for $w_i$ is $5r/2 + k - i/2$, which is at least $2r + k$, since $i \leq r$.

If $r \leq i \leq 2r - k$, then the yield for $w_i$ is $2r + k$.

If $2r - k \leq i \leq \lfloor 3r/2 \rfloor$, then the yield for $w_i$ is also $2r + k$.
\end{proof}

\begin{theorem}
For $1 \leq s \leq 2$, SPEEDUPW12 finds an $8\gamma/(s+1)$-approximation to the repairman problem on windows in set $W_2$ in $O(\min\{r, m\}\Gamma(n))$ time.
\end{theorem}

\begin{proof}
By Lemma \ref{lemma:lowestinterval W_2}, our analysis gives no yield less than $2r + k$.  Since we use two copies each of $A$ and $A^R$ and a single copy each of $A_{r - k}$, $A_{r - k}^R$, $A_{2r - k}$, and $A_{2r - k}^R$, averaged over $r$ different sets of periods, we apply the Average Coverage Proposition over $8r$ runs.  Thus, the fraction of optimal profit obtained is ${(2r + k)/( 8\gamma r)} = {((r + k) + r)/(8\gamma r)} = {(s + 1)/(8\gamma)}$.
\end{proof}

\subsection{Speedup $2 \leq s \leq 3$ for Windows in Set $W_2$}

For the range $2 \leq s \leq 5/2$, we can represent any rational speedup $s$ in the form $s = (2r + k)/r$ with integers $r \geq 1$ and $0 \leq k \leq r - k$.  For this analysis, we consider service runs $A$, $A^R$, $A_{r - 2k}$, and $A_{r - 2k}^R$, noting that $\lambda = 2$.  We will use three copies each of $A$ and $A^R$ and a single copy each of $A_{r - 2k}$ and $A_{r - 2k}^R$.  Because the generation of the tables and the case analysis needed to show the coverage are involved and of a similar form as Lemma \ref{lemma:lowestinterval W_2}, we have moved these details to Appendix \ref{appendix:2<s<3 W_2}.

\begin{theorem}
For $2 \leq s \leq 5/2$, SPEEDUPW12 finds an $8\gamma/(2s - 1)$-approximation to the repairman problem on windows in set $W_2$ in $O(\min\{r, m\}\Gamma(n))$ time.
\end{theorem}

\begin{proof}
By Lemma \ref{lemma:lowestinterval 2<s<5/2 W_2}, the yield is at least $3r + 2k$.  Since three copies each of $A$ and $A^R$ and a single copy each of $A_{r - 2k}$ and $A_{r - 2k}^R$ are used, averaged over $r$ different sets of periods, the Average Coverage Proposition is applied over $8r$ runs.  Thus, the fraction of optimal profit obtained is at least ${(3r + 2k)/(8\gamma r) } = {((4r + 2k) - r)/(8\gamma r)} = {(2s - 1)/(8\gamma)}$.
\end{proof}

\begin{observation}
\label{observation:speedup}
Given a speedup $s' > s$, we can always simulate with speedup $s'$ the runs used in the analysis of speedup $s$ by introducing delays at appropriate points in each run.  Thus, an approximation ratio of $\beta$ at speedup $s$ is an upper bound on the approximation ratio at speedup $s'$.
\end{observation}

By Observation \ref{observation:speedup}, the $2\gamma$-approximation for $s = 5/2$ implies at most a constant $2\gamma$-approximation to the repairman problem on windows in set $W_2$ when $5/2 \leq s \leq 3$.

\subsection{Speedup $3 \leq s \leq 5$ for Windows in Set $W_2$}

For set $W_2$ with $3 \leq s \leq 5$ where $s = q/r$, we consider runs $A$, $A^R$, and their shifts, noting that $\lambda = 2$.
When $s < 4$, runs $A$ and $\vec{A}$ give full coverage for service requests in subsets of $S^0$ and $S^1$ and partial coverage of service requests in $S^2$, while runs $A^R$ and $\vec{A}^R$ give full coverage for service requests in subsets of $S^3$ and $S^2$ and partial coverage of $S^1$.  When $s = 4$, runs $A$ and $\vec{A}$ go further by also giving full coverage for service requests in subsets of $S^2$, while runs $A^R$ and $\vec{A}^R$ also give full coverage for service requests in subsets of $S^1$.  When $s > 4$, runs $A$ and $\vec{A}$ give full coverage for service requests in subsets of $S^0$, $S^1$, and $S^2$ and partial coverage of service requests in $S^3$, while runs $A^R$ and $\vec{A}^R$ give full coverage for service requests in subsets of $S^3$, $S^2$ and $S^1$ and partial coverage of $S^0$.  Since the contributions of the $A$ and $A^R$ runs and their shifted versions tend to balance each other, we can analyze this balance between the two over all possible sets of periods to find a lower bound on the total profit after trimming.

\begin{theorem}
For $3 \leq s \leq 5$, SPEEDUPW12 finds a $4\gamma /(s - 1)$-approximation for windows $W_2$ in $O(\min\{r, m\}\Gamma(n))$ time.
\end{theorem}

\begin{proof}
For runs $A$ and $\vec{A}$, $w_i$ earns a 1 (denoting full coverage) for each of the $r$ sets of periods where $0 \leq i \leq q - 2r$, giving a total of $r$ for each such $i$.  For each $i > q - 2r$, the total decreases by $1/2$ from the total for $i - 1$.   For runs $A^R$ and $\vec{A}^R$, $w_i$ gets 1 for each of the $r$ sets of periods where $3r - (q - 2r) \leq i \leq 3r$, giving a total of $r$ for each such $i$.  For each $i < 3r - (q - 2r)$, the total decreases by $1/2$ from the total for $i + 1$.

The combined contributions of runs $A$, $A^R$, and their shifted versions is $q/2 - r/2$ for $w_0$.  Contributions from run $A$ are constant, and contributions from run $A^R$ only increase or stay constant for $w_i$ where $0 < i \leq r$.  Contributions for $w_i$ for all runs sum to $2r$ for $r < i < \lfloor 3r/2 \rfloor$.  Thus, the yield for all $w_i$ is at least $q/2 - r/2$.  Since two runs averaged over $r$ different sets of periods are used, the fraction of optimal profit obtained is at least $(q/2 - r/2)\cdot 1/(2\gamma r) = q/(4\gamma r) - r/(4\gamma r) = (s - 1)/(4\gamma)$.
\end{proof}

\section{New Types of Runs to Handle Windows in Set $W_3$}
\label{section:runs for W_3}

In addition to type $A$ runs, which repeat every 2 periods, our analysis of windows in $W_3$ defines type $B$ and type $C$ runs which repeat every 3 or every 4 periods, respectively.  We define $B$ and $C$ runs only in the range $2 < s \leq 3$.  Let $\nu = \lceil s \rceil - s$.

Start run $B$ at $t = 0$ at the location that $R^*$ has at time $t = -0.5$.  From there, run $B$ follows a repeating pattern of racing forward along $R^*$ for $2$ periods, racing backward along $R^*$ for $1 - \nu/(2s)$ periods, and racing forward along $R^*$ for $\nu/(2s)$ periods.  As for run $C$, also start it at $t = 0$ at the location that $R^*$ has at time $t = -0.5$.  From there, run $C$ follows a repeating pattern of racing forward along $R^*$ for $2 + 2/s$ periods and racing backward along $R^*$ for $2 - 2/s$ periods.

Similar to $A^R$, we also define $B^R$ and $C^R$, the ``reverses'' of runs $B$ and $C$, respectively.  Both runs $B^R$ and $C^R$ start at $t = 0$ at the location that $R^*$ has at time $t = \lambda/2$.  From its starting point, run $B^R$ follows a repeating pattern of racing forward along $R^*$ for $\nu/(2s)$, racing backward along $R^*$ for $1 - \nu/(2s)$ periods, and racing forward along $R^*$ for $2$ periods.  From its starting point, run $C^R$ follows a repeating pattern of racing backward along $R^*$ for $2 - 2/s$ periods and forward along $R^*$ for $2 + 2/s$ periods.  As with $A$, let $B_\Delta$ and $C_\Delta$, respectively, be runs $B$ and $C$ moved forward $\Delta$ hops, and let runs $B_\Delta^R$ and $C_\Delta^R$, respectively, be runs $B^R$ and $C^R$ moved backward $\Delta$ hops.  Runs $B_\Delta$ and $C_\Delta$, respectively, follow the same patterns of movement as runs $B$ and $C$ but start at $t = 0$ at the location that $R^*$ has at  $t = -0.5 + \Delta/(2r)$.  Their reverses $B_\Delta^R$ and $C_\Delta^R$, respectively, follow the same patterns of movement as $B^R$ and $C^R$ but start at $t = 0$ at the location that $R^*$ has at $t = \lambda/2 - \Delta/(2r)$.

As there is for type $A$ runs, there are unique patterns of coverage corresponding to type $B$ and $C$ runs and their reverses.   Recall that $s = q/r$.  Note that, for analysis of $B$ and $C$ runs, we choose the smallest values of $q$ and $r$ such that $q + r$ is even.  Since the number of subsets is determined by $r$, it is necessary for $B$ and $C$ runs to have an even $q + r$ in order to keep the coverage defined in terms of complete rather than partial subsets.  A type $B$ run moves forward during its first two periods of time the same distance that an optimal run moves during $2q$ subintervals.  During its third period of time, it moves backward the same distance that an optimal run moves during $q( 1 - \nu/(2s))$ subintervals and then forward the same distance that an optimal run moves during $q\nu/(2s)$ subintervals.  Then, the pattern repeats.

We recall the list of subsets: $S^0_1$, $S^0_2$ $\ldots$ $S^0_r$, $S^1_1$ $\ldots$ $S^1_r$, $S^2_1$ $\ldots$ $S^2_r$, $S^3_1$ $\ldots$ $S^3_r$, $S^4_1$ $\ldots$ $S^4_r$.
When $2 < s \leq 3$, run $B$, during the first period in its pattern, covers $q$ successive subsets as it moves forward.  In the second period, it covers another $q$ subsets moving forward.  Finally, in its third period, it covers $(3q - 3r)/2$ subsets backward but no new subsets forward, since the subsets covered forward were already covered backward.    We see that only the first period in the pattern covers the first $(q - r)/2$ subsets.  Then the first and third period in the pattern cover the next $(q - r)/2$ subsets.  All three periods cover the next $r$ subsets.  The second and third periods cover the next $q - 2r$ subsets, and only the second period covers the final $r$ subsets.  This pattern of coverage is represented as $(q - r)/2$ repetitions of 1/3, $(q - r)/2$ repetitions of 2/3, $r$ repetitions of 1, $q - 2r$ repetitions of 2/3, and $r$ repetitions of 1/3.

Figure \ref{figure:bspeeds} gives two examples of type $B$ runs for speedups in the range $17/7 \leq s \leq 3$, the only range for which our analysis will employ type $B$ runs.  Observe that the run for $s = 5/2$ uses the form $10/4$ in order to conform with the restriction for our analysis that $q + r$ must be even.  Unlike $A$ runs which repeat every two periods, both the runs in this figure arrive at the same corresponding position at the beginning of every third period, namely at times $0, 1.5, 3,$ and so on.  Portions of runs servicing requests in $S^0$, $S^1$, $S^2$, $S^3$, and $S^4$ or various subsets are identified: subsets $S^0_4$, $S^1_3$, $S^1_4$, $S^2_1$, and $S^2_2$  mapping to quarter periods of $R^*$ for $s = 10/4$ and subsets $S^0_5$, $S^1_4$, $S^1_5$, $S^2_1$, $S^2_2$, $S^3_3$, and $S^4_1$ mapping to a fifth of a period of $R^*$ for $s = 13/5$.  Focusing on the example of $s = 10/4$ where $q = 10$ and $r = 4$, note that, during a three-period section, subsets $S^0_1$ through $S^0_3$ are covered a single time, subsets $S^0_4$, $S^1_1$, and $S^1_2$ are covered twice, subsets $S^1_3$, $S^1_4$, $S^2_1$, and $S^2_2$ are covered all three times, subsets $S^2_3$ and $S^2_4$ are covered twice, and subsets $S^3_1$ through $S^3_4$ are covered a single time.  This pattern of 3 repetitions of 1/3, 3 repetitions of 2/3, 4 repetitions of 1, 2 repetitions of 2/3, and 4 repetitions of 1/3 exactly corresponds to the repeating pattern of subset coverage described in the previous paragraph.

\begin{figure}[!hbt]
\centering
\begin{xy}
\newxycolor{white}{1 gray}
\xyimport(533, 208){\includegraphics[width=.80\textwidth]{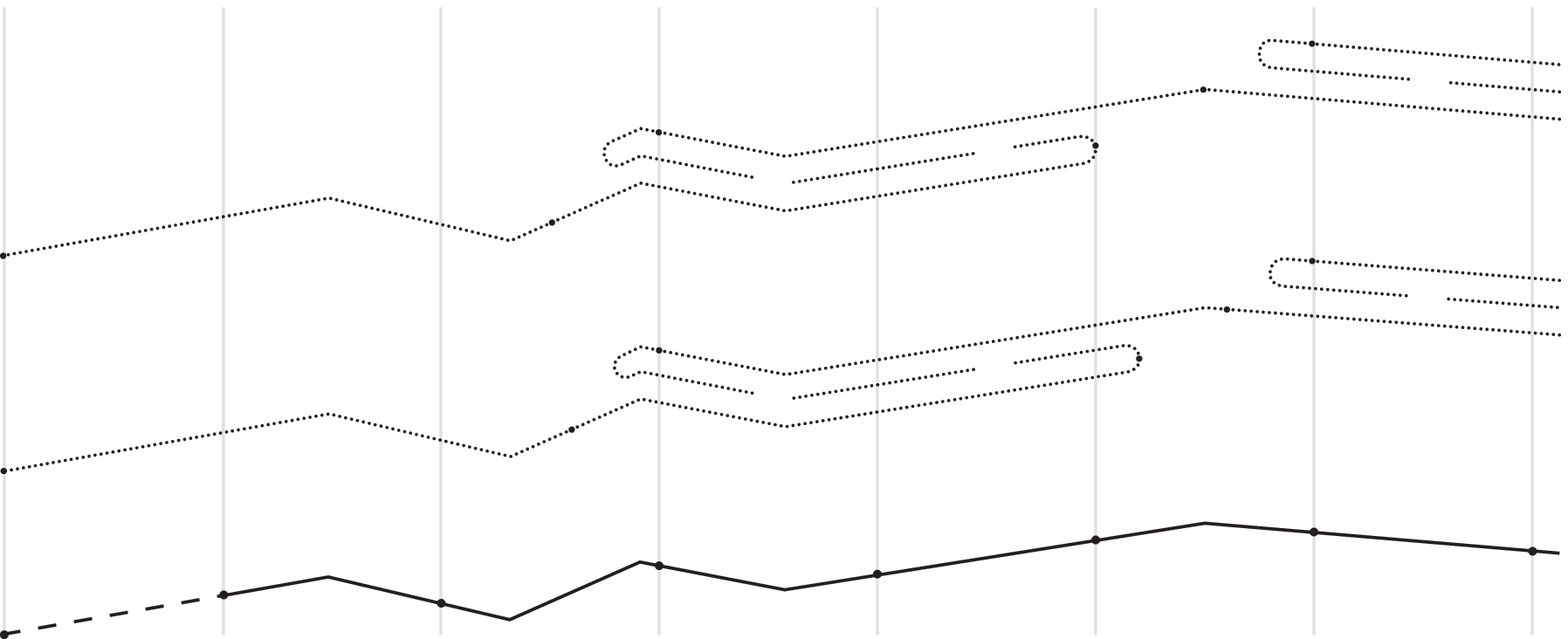}}
(-24,135)*!R\txt\footnotesize{Run $B$ for};
(-24,117)*!R{\large s = {5\over 2} = {10 \over 4}};
(-24,62)*!R\txt\footnotesize{Run $B$ for};
(-24,44)*!R{\large s = {13 \over 5}};
(-24,8)*!R\txt\footnotesize{Optimal};
(-24,-4)*!R\txt\footnotesize{Run $R^*$};
%% Optimal Times
(-6,-10)*\txt\footnotesize{$-.5$};
(76,4)*\txt\footnotesize{$0$};
(149,1)*\txt\footnotesize{$.5$};
(225,13)*\txt\footnotesize{$1$};
(300,10)*\txt\footnotesize{$1.5$};
(374,23)*\txt\footnotesize{$2$};
(449,25)*\txt\footnotesize{$2.5$};
(523,19)*\txt\footnotesize{$3$};
%% 10/4 times
(1,116)*\txt\scriptsize{$0$};
(187,146)*\txt\scriptsize{$.5$};
(383,164)*\txt\scriptsize{$1$};
(225,176)*\txt\scriptsize{$1.5$};
(411,190)*\txt\scriptsize{$2$};
(448,205)*\txt\scriptsize{$3$};
%% 13/5 times
(1,45)*\txt\scriptsize{$0$};
(193,78)*\txt\scriptsize{$.5$};
(397,94)*\txt\scriptsize{$1$};
(225,105)*\txt\scriptsize{$1.5$};
(419,118)*\txt\scriptsize{$2$};
(448,133)*\txt\scriptsize{$3$};
%% 10/4 Labels
(40,142)*\txt\scriptsize{$S^0$};
(115,154)*\txt\scriptsize{$S^1$};
(197,131)*\txt\tiny{$S^1_3$};
(213,137)*\txt\tiny{$S^1_4$};
(160,124)*\txt\tiny{$S^2_1$};
(179,124)*\txt\tiny{$S^2_2$};
(210,175)*\txt\tiny{$S^0_4$};
(263,133)*\txt\scriptsize{$S^2$};
(263,151)*\txt\scriptsize{$S^1$};
(263,169)*\txt\scriptsize{$S^0$};
(339,144)*\txt\scriptsize{$S^3$};
(339,160)*\txt\scriptsize{$S^2$};
(339,177)*\txt\scriptsize{$S^1$};
(423,171)*\txt\tiny{$S^1_3$};
(438,169)*\txt\tiny{$S^1_4$};
(384,185)*\txt\tiny{$S^2_1$};
(399,187)*\txt\tiny{$S^2_2$};
(437,206)*\txt\tiny{$S^0_4$};
(488,166)*\txt\scriptsize{$S^2$};
(488,184)*\txt\scriptsize{$S^1$};
(488,201)*\txt\scriptsize{$S^0$};
%% 13/5 Labels
(40,71)*\txt\scriptsize{$S^0$};
(115,83)*\txt\scriptsize{$S^1$};
(263,62)*\txt\scriptsize{$S^2$};
(263,80)*\txt\scriptsize{$S^1$};
(263,97)*\txt\scriptsize{$S^0$};
(339,73)*\txt\scriptsize{$S^3$};
(339,89)*\txt\scriptsize{$S^2$};
(339,106)*\txt\scriptsize{$S^1$};
(203,62)*\txt\tiny{$S^1_4$};
(218,68)*\txt\tiny{$S^1_5$};
(158,54)*\txt\tiny{$S^2_1$};
(173,51)*\txt\tiny{$S^2_2$};
(188,55)*\txt\tiny{$S^2_3$};
(210,103)*\txt\tiny{$S^0_5$};
(383,79)*\txt\tiny{$S^4_1$};
(426,99)*\txt\tiny{$S^1_4$};
(440,97)*\txt\tiny{$S^1_5$};
(382,113)*\txt\tiny{$S^2_1$};
(395,115)*\txt\tiny{$S^2_2$};
(408,117)*\txt\tiny{$S^2_3$};
(437,134)*\txt\tiny{$S^0_5$};
(488,95)*\txt\scriptsize{$S^2$};
(488,113)*\txt\scriptsize{$S^1$};
(488,130)*\txt\scriptsize{$S^0$};
\end{xy}
\caption{Examples of type $B$ runs for two different speedups in the range $17/7 \leq s \leq 3$, namely at $5/2$ and $13/5$.}
\label{figure:bspeeds}
\end{figure}

A type $C$ run moves forward during its first two periods of time the same distance that an optimal run moves during $2q$ subintervals.  During its third period of time, it moves forward the same distance that an optimal run moves during $2q/s$ subintervals and then backward the same distance that an optimal run moves during $(1 - 2/s)q$ subintervals.  During its fourth period of time, it moves backward the same distance that an optimal run moves during $q$ subintervals.  Then, the pattern repeats.  

When $2 \leq s \leq 5/2$, run $C$, during the first period in its pattern, covers $q$ successive subsets as it moves forward.  In its second period, it covers another $q$ subsets moving forward.  In its third period, it covers $2r$ subsets forward but no new subsets backward.  Finally, in its fourth period, it covers $q$ subsets moving backward.  We see that only the first period in the pattern covers the first $r$ subsets.  Then, the first and fourth period in the pattern cover the next $q - 2r$ subsets.  The first, second, and fourth periods cover the next $r$ subsets.  The second and fourth periods cover the next $q - 2r$ subsets.  The second, third, and fourth periods cover the next $3r - q$ subsets.  Only the second and the third period cover the next $q - 2r$ subsets, and only the third period covers the final $r$ subsets.  This pattern of coverage is represented as $r$ repetitions of 1/4, $q - 2r$ repetitions of 1/2, $r$ repetitions of 3/4, $q - 2r$ repetitions of 1/2, $3r - q$ repetitions of 3/4, $q - 2r$ repetitions of 1/2, and $r$ repetitions of 1/4.  Just as with type $A$ runs, we will use the patterns for $B$ and $C$ runs in conjunction with CREATE-TABLE-$\lambda$ to construct tables showing bounds on average coverage.

Figure \ref{figure:cspeeds} gives two examples of type $C$ runs for speedups in the range $2 \leq s \leq 17/7$, the only range for which our analysis will employ type $C$ runs.  Observe that the run for $s = 2$ uses the form $4/2$ in order to conform with the restriction for our analysis that $q + r$ must be even.  Unlike $A$ and $B$ runs with, respectively, repeating patterns of two and three periods, both the runs in this figure arrive at the same position at the beginning of every fourth period, namely at times $0$, $2$, and so on.  Portions of runs servicing requests in $S^0$, $S^1$, $S^2$, $S^3$, and $S^4$ or various subsets are identified: no separate subsets for $s = 2$ but subsets $S^1_2$, $S^1_3$, $S^1_4$, $S^1_5$, $S^2_1$, $S^3_1$, $S^3_2$, $S^4_1$, and $S^4_2$ mapping to a fifth of a period of $R^*$ for $s = 11/5$.  Focusing on the example of $s = 11/5$ where $q = 11$ and $r = 5$, note that, during a four-period section, subsets $S^0_1$ through $S^0_5$ are covered a single time, subset $S^1_1$ is covered twice, subsets $S^1_2$ through $S^1_5$ and $S^2_1$ are covered three times, subset $S^2_2$ is covered twice, subsets $S^2_3$ through $S^2_5$ and $S^3_1$ are covered three times, subset $S^3_2$ is covered twice, and subsets $S^3_3$ through $S^3_5$ and $S^4_1$ and $S^4_2$ are covered a single time.  This pattern of 5 repetitions of 1/4, 1 repetition of 1/2, 5 repetitions of 3/4, 1 repetition of 1/2, 4 repetitions of 3/4, 1 repetition of 1/2, and 5 repetitions of 1/4 exactly corresponds to the repeating pattern of subset coverage described in the previous paragraph.

\begin{figure}[!hbt]
\centering
\begin{xy}
\newxycolor{white}{1 gray}
\xyimport(534, 210){\includegraphics[width=.80\textwidth]{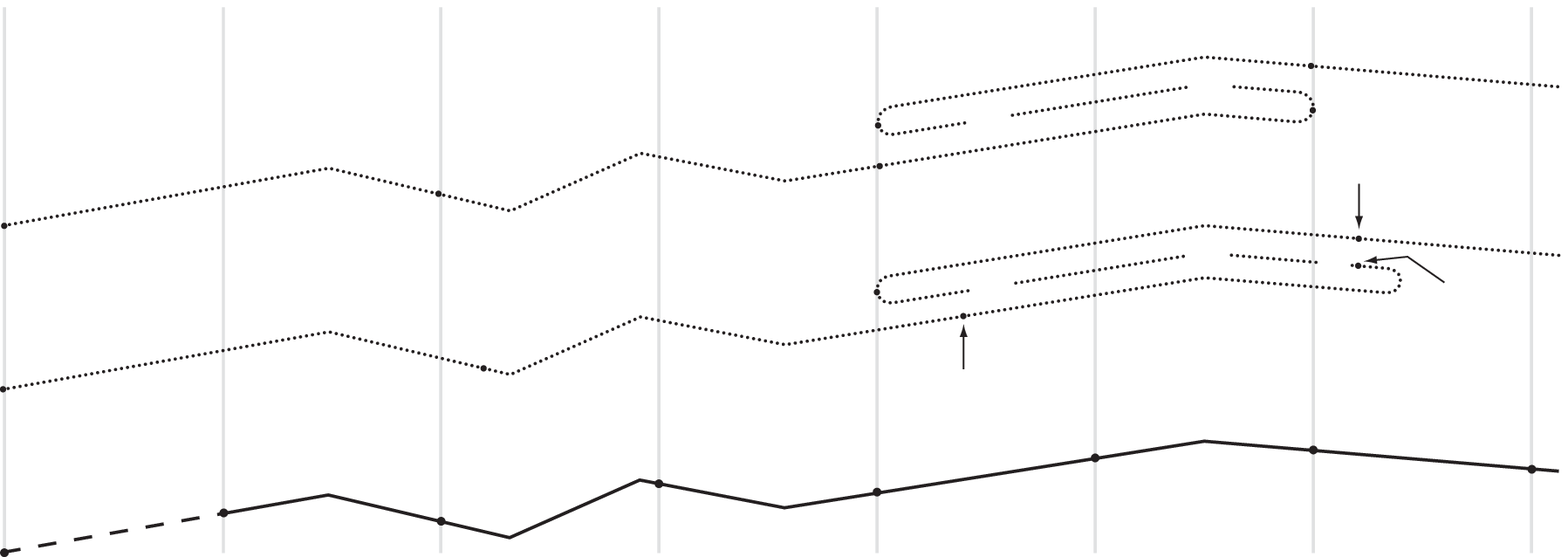}}
(-24,137)*!R\txt\footnotesize{Run $C$ for};
(-24,119)*!R{\large s = 2 = {4 \over 2}};
(-24,71)*!R\txt\footnotesize{Run $C$ for};
(-24,53)*!R{\large s = {11 \over 5}};
(-24,8)*!R\txt\footnotesize{Optimal};
(-24,-4)*!R\txt\footnotesize{Run $R^*$};
%% Optimal Times
(-6,-11)*\txt\footnotesize{$-.5$};
(76,3)*\txt\footnotesize{$0$};
(149,0)*\txt\footnotesize{$.5$};
(225,12)*\txt\footnotesize{$1$};
(300,9)*\txt\footnotesize{$1.5$};
(374,22)*\txt\footnotesize{$2$};
(449,24)*\txt\footnotesize{$2.5$};
(523,18)*\txt\footnotesize{$3$};
%% 4/2 times
(1,116)*\txt\scriptsize{$0$};
(149,128)*\txt\scriptsize{$.5$};
(301,139)*\txt\scriptsize{$1$};
(461,170)*\txt\scriptsize{$1.5$};
(292,165)*\txt\scriptsize{$2$};
(449,199)*\txt\scriptsize{$2.5$};
%% 11/5 times
(1,54)*\txt\scriptsize{$0$};
(164,62)*\txt\scriptsize{$.5$};
(329,64)*\txt\scriptsize{$1$};
(496,98)*\txt\scriptsize{$1.5$};
(292,101)*\txt\scriptsize{$2$};
(465,149)*\txt\scriptsize{$2.5$};
%% 4/2 Labels
(40,145)*\txt\scriptsize{$S^0$};
(115,159)*\txt\scriptsize{$S^1$};
(190,151)*\txt\scriptsize{$S^1$};
(265,157)*\txt\scriptsize{$S^2$};
(339,148)*\txt\scriptsize{$S^2$};
(339,169)*\txt\scriptsize{$S^1$};
(339,188)*\txt\scriptsize{$S^0$};
(414,160)*\txt\scriptsize{$S^3$};
(414,181)*\txt\scriptsize{$S^2$};
(414,201)*\txt\scriptsize{$S^1$};
(489,195)*\txt\scriptsize{$S^2$};
%% 11/5 Labels
(40,81)*\txt\scriptsize{$S^0$};
(115,95)*\txt\scriptsize{$S^1$};
(159,84)*\txt\tiny{$S^2_1$};
(174,83)*\txt\tiny{$S^1_2$};
(189,88)*\txt\tiny{$S^1_3$};
(204,96)*\txt\tiny{$S^1_4$};
(219,102)*\txt\tiny{$S^1_5$};
(265,93)*\txt\scriptsize{$S^2$};
(308,78)*\txt\tiny{$S^3_1$};
(322,81)*\txt\tiny{$S^3_2$};
(339,84)*\txt\tiny{$S^2_3$};
(353,86)*\txt\tiny{$S^2_4$};
(367,88)*\txt\tiny{$S^2_5$};
(339,105)*\txt\scriptsize{$S^1$};
(339,124)*\txt\scriptsize{$S^0$};
(414,97)*\txt\scriptsize{$S^3$};
(414,117)*\txt\scriptsize{$S^2$};
(414,137)*\txt\scriptsize{$S^1$};
(457,92)*\txt\tiny{$S^4_1$};
(471,91)*\txt\tiny{$S^4_2$};
(457,112)*\txt\tiny{$S^3_1$};
(457,132)*\txt\tiny{$S^2_1$};
(471,131)*\txt\tiny{$S^1_2$};
(486,129)*\txt\tiny{$S^1_3$};
(501,128)*\txt\tiny{$S^1_4$};
(516,126)*\txt\tiny{$S^1_5$};
\end{xy}
\caption{Examples of type $C$ runs for two different speedups in the range $2 \leq s \leq 17/7$, namely at $2$ and $11/5$.}
\label{figure:cspeeds}
\end{figure}

\section{Speedup Performance for Windows in Set $W_3$}
\label{section:performance W_3}

We will now explore the speedup-performance trade-off for windows in $W_3$ for all speedups $1 \leq s \leq 6$.  For set $W_3$, our analysis must consider subsets $w_0$ through $w_{4r}$.  As before, we will assign a 1 for full coverage and a $1/2$ for half coverage of any subset.  Because of the coverage patterns of $B$ and $C$ runs, we will also assign values of $1/4$, $1/3$, $2/3$, and $3/4$ for corresponding proportions of coverage.  For the subsets for a given range of speedup values for $W_3$, the values are symmetrical around $w_{2r}$.  Thus, our tables and proofs will not list contributions for subset $w_i$ where $i > 2r$.

\subsection{Speedup $1 \leq s \leq 2$ for Windows in Set $W_3$}

For this analysis, we consider service runs $A$, $A^R$, $A_{r - k}$, $A_{r - k}^R$, $A_{2r - k}$, $A_{2r - k}^R$, $A_{3r - k}$, and $A_{3r - k}^R$, noting that $\lambda = 3$.  We will use two copies each of $A$ and $A^R$ and a single copy each of $A_{r - k}$, $A_{r - k}^R$, $A_{2r - k}$, $A_{2r - k}^R$, $A_{3r - k}$, and $A_{3r - k}^R$.  We have moved the generation of the tables and the case analysis needed to show the coverage to Appendix \ref{appendix:1<s<2 W_3}.

\begin{theorem}
For $1 \leq s \leq 2$, our algorithm finds a $10\gamma /(s+1)$-approximation to the repairman problem on windows in set $W_3$ in $O(\min\{r, m\}\Gamma(n))$ time.
\end{theorem}

\begin{proof}
By Lemma \ref{lemma:lowestinterval W_3}, our analysis gives no yield less than $2r + k$.  Since two copies each of $A$ and $A^R$ and a single copy each of $A_{r - k}$, $A_{r - k}^R$, $A_{2r - k}$, $A_{2r - k}^R$, $A_{3r - k}$, and $A_{3r - k}^R$ are used, averaged over $r$ different sets of periods, the Average Coverage Proposition is applied over $10r$ runs.  Thus, the fraction of optimal profit obtained is at least ${(2r + k)/(10\gamma r)} = {((r + k) + r)/(10\gamma r)} = {(s + 1)/(10\gamma)}$.
\end{proof}

\subsection{Speedup $2 \leq s \leq 7/3$ for Windows in Set $W_3$}

For the range $2 \leq s \leq 7/3$, any rational speedup $s$ can be represented in the form $s = (2r + k)/r$ with integers $r \geq 1$ and $0 \leq k \leq r/3$.  For this analysis, we consider service runs $A$, $A^R$, $C_{(3r - k)/2}$, and $C_{(3r - k)/2}^R$, noting that $\lambda = 3$.  We will use a single copy each of $A$ and $A^R$ and two copies each of $C_{(3r - k)/2}$ and $C_{(3r - k)/2}^R$.  We have moved the  generation of the tables and the case analysis needed to show the coverage to Appendix \ref{appendix:2<s<7/3 W_3}.

\begin{theorem}
For $2 \leq s \leq 7/3$, algorithm SPEEDUPW12 finds a $6\gamma/s$-approximation to the repairman problem on windows in set $W_3$ in $O(\min\{r, m\}\Gamma(n))$ time.
\end{theorem}

\begin{proof}
By Lemma \ref{lemma:lowestinterval 2<s<7/3 W_3}, our analysis gives no yield less than $2r + k$.  Since 1 copy of each of $A$ and $A^R$ and 2 copies each of $C_{(3r - k)/2}$ and $C_{(3r - k)/2}^R$ were used, averaged over $r$ different sets of periods, the Average Coverage Proposition is applied over $6r$ runs.  Thus, the fraction of optimal profit obtained is at least ${ (2r + k)/(6\gamma  r) } = {s/(6\gamma) }$.
\end{proof}

\subsection{Speedup $7/3 \leq s \leq 17/7$ for Windows in Set $W_3$}

For the range $7/3 \leq s \leq 17/7$, any rational speedup $s$ can be represented in the form $s = (2r + k)/r$ with integers $r \geq 1$ and $r/3 \leq k \leq 3r/7$.  For this analysis, we consider service runs $A$, $A^R$, $C_{2r - 2k}$, and $C_{2r - 2k}^R$, noting that $\lambda = 3$.  We will use $k$ copies each of $A$ and $A^R$ and $r - k$ copies each of $C_{2r - 2k}$ and $C_{2r - 2k}^R$.  We have moved generation of the tables and the case analysis needed to show the coverage to Appendix \ref{appendix:7/3<s<17/7 W_3}.

\begin{theorem}
For $7/3 \leq s \leq 17/7$, algorithm SPEEDUPW12 finds an $8\gamma/(s^2 - 4s + 7)$-approximation to the repairman problem on windows in set $W_3$ in $O(\min\{r, m\}\Gamma(n))$ time.
\end{theorem}

\begin{proof}
By Lemma \ref{lemma:lowestinterval 7/3<s<17/7 W_3}, our analysis gives no yield less than $3r^2/4 + k^2/4$.  Since $k$ copies of each of $A$ and $A^R$ and $r - k$ copies each of $C_{2r - 2k}$ and $C_{2r - 2k}^R$ were used, averaged over $r$ different sets of periods, the Average Coverage Proposition is applied over $2r^2$ runs.  Thus, the fraction of optimal profit obtained is at least ${ (3r^2/4 + k^2/4)/(2\gamma r^2) } = {((2r + k)^2 - 4r(2r + k) + 7r^2)/(8\gamma r^2)} = ( s^2 - 4s + 7)/(8\gamma)$.
\end{proof}

\subsection{Speedup $17/7 \leq s \leq 3$ for Windows in Set $W_3$}

For the range $17/7 \leq s \leq 3$, any rational speedup $s$ can be represented in the form $s = (2r + k)/r$ with integers $r \geq 1$ and $3r/7 \leq k \leq r$.  For this analysis, we consider service runs $A$, $A^R$, $B_{r - k + 1}$, and $B_{r - k + 1}^R$, noting that $\lambda = 3$. We will use $6r - 4k$ copies each of $A$ and $A^R$ and $3r - 3k$ copies each of $B_{r - k + 1}$ and $B_{r - k + 1}^R$.  We have moved the generation of the tables and the case analysis needed to show the coverage to Appendix \ref{appendix:17/7<s<3 W_3}.

\begin{theorem}
For $17/7 \leq s \leq 3$, algorithm SPEEDUPW12 finds a $\gamma(23 - 7s)/(1 + 3s - s^2)$-approximation to the repairman problem on windows in set $W_3$ in $O(\min\{r, m\}\Gamma(n))$ time.
\end{theorem}

\begin{proof}
By Lemma \ref{lemma:lowestinterval 17/7<s<3 W_3}, our analysis gives no yield less than $6r^2 - 2rk - 2k^2$.  Since $6r - 4k$ copies of each of $A$ and $A^R$ and $3r - 3k$ copies each of $B_{r - k + 1}$ and $B_{r - k + 1}^R$ were used, averaged over $r$ different sets of periods, the Average Coverage Proposition is applied over $18r^2 - 14rk$ runs.  Thus, the fraction of optimal profit obtained is at least ${ (6r^2 - 2rk - 2k^2)/(\gamma (18r^2 - 14rk)) } = {(r^2 + 3r(2r + k) - (2r + k)^2)/(\gamma r(23r - 7(2r + k)))} = (1 + 3 s - s^2)/(\gamma(23 - 7s))$.
\end{proof}

By Observation \ref{observation:speedup}, the $2\gamma$-approximation for $s = 3$ implies at most a constant $2\gamma$-approximation to the repairman problem on windows in set $W_3$ when $3 \leq s \leq 4$.

\subsection{Speedup $4 \leq s \leq 6$ for Windows in Set $W_3$}

For set $W_3$ with $4 \leq s \leq 6$ where $s = q/r$, we consider runs $A$ and $A^R$, noting that $\lambda = 3$.
When $s < 5$, run $A$ and its shift give full coverage in subsets of $S^0$, $S^1$, and $S^2$, and partial coverage in subsets of $S^3$ and $S^4$, while run $A^R$ and its shift give full coverage in subsets of $S^4$, $S^3$, and $S^2$ and partial coverage in subsets of $S^1$ and $S^0$.  When $s \geq 5$, runs $A$ and $\vec{A}$ give full coverage in subsets of $S^0$, $S^1$, $S^2$, and $S^3$ and partial coverage in subsets of $S^4$, while runs $A^R$ and $\vec{A}^R$ give full coverage in subsets of $S^4$, $S^3$, $S^2$, and $S^1$ and partial coverage in subsets of $S^1$ and $S^0$.  Since the contributions of the $A$ and $A^R$ runs and their shifted versions tend to balance each other, we can analyze this balance between the two over all possible sets of periods to find a lower bound on the total profit after trimming.

\begin{theorem}
For $4 \leq s \leq 6$, our algorithm finds a $4\gamma/(s - 2)$-approximation for windows $W_3$ in $O(\min\{r, m\}\Gamma(n))$ time.
\end{theorem}

\begin{proof}
A 1 is assigned for any subset which is covered every period, and a $1/2$ is assigned for any subset covered every other period.  For runs $A$ and $\vec{A}$, $w_i$ earns a 1 for all $r$ sets of periods where $0 \leq i \leq q - 2r$, giving a total of $r$ for each such $i$.  For each $i > q - 2r$, the total decreases by $1/2$ from the total for $i - 1$.   For runs $A^R$ and $\vec{A}^R$, $w_i$ gets 1 for all $r$ sets of periods where $4r - (q - 2r) \leq i \leq 4r$, giving a total of $r$ for each such $i$.  For each $i < 4r - (q - 2r)$, the total decreases by $1/2$ from the total for $i + 1$.

Now, we take the total over all sets of periods for both runs.  For runs $A$ and $\vec{A}$, we get a total of $r$ for $w_0$.  For runs $A^R$ and $\vec{A}^R$, we get a total of $r - (1/2)(4r - (q - 2r)) = q/2 - 2r$.  Summing these together, we get a yield of $q/2 - r$ for $w_0$.  By symmetry, the total for $w_{4r}$ is also $q/2 - r$.  Contributions from $A$ and $\vec{A}$ are constant and contributions from $A^R$ and $\vec{A}^R$ only increase or stay constant for $0 < i \leq r$.  Contributions for $w_i$ for all runs sum to $2r$ for $r < i < 3r$.   Thus, the yield for all other $w_i$ in all cases is at least $q/2 - r$.  Since $r$ sets of periods for the two pairs of runs cost a total of $2r$ sets of periods to average over, the fraction of profit after trimming is at least $(q/2 - r)\cdot 1/(2r) = q/(4r) - 2r/(4r) = (s - 2)/4$.  Multiplying the reciprocal by $\gamma$ gives a $4\gamma/(s - 2)$-approximation.
\end{proof}

\section{Performance of SPEEDUPW12}
\label{section:performance of speedupw12}

Now that we have characterized the performance of our speedup algorithms on windows in sets $W_1$, $W_2$, and $W_3$, we bound the performance of SPEEDUPW12 by combining our results as follows.
Let $R^*$ be an optimal service run for a repairman instance with time window lengths from 1 up to 2.
Consider a new set of periods with duration $.25$. Partition windows into the sets $H_3$, $H_4$, $H_5$, $H_6$, and $H_7$, such that for $i = 3, 4, 5, 6, 7$, a window is put in $H_i$ if it completely contains exactly $i$ of these new periods of length .25.  Let the total fraction of profit in an optimal solution coming from windows in set $H_\ell$ be $h_\ell$.  Thus, $\sum_{\ell = 3}^7 h_\ell = 1$.

We use these subintervals when analyzing the performance of the algorithm run on periods of length $.5$, $.75$, and $1$.  Consider set $H_\ell$ of windows, $\ell = 3, 4, 5, 6, 7$ and period length $j/4$, for $j = 1, 2, 3, 4$.
The number of full subintervals of a window in $H_\ell$ that are covered when the period length is $j/4$ is either
$\lfloor (\ell - j + 1)/j\rfloor$ or $\lceil (\ell - j + 1)/j\rceil$
depending on which set of periods is used.

For $\ell = 1, 2, 3$, let $f_\ell(s)/\gamma$ be the fraction of optimal profit earned for SPEEDUPW12
applied to requests with windows in $W_\ell$.  Recall that coverage of windows in $W_\ell$ is defined using period size .5.  To apply this coverage to the three period sizes used in the algorithm, we establish:

\begin{lemma}
\label{lemma:phase1}
The first phase of SPEEDUPW12 yields a run $\hat{R}$ with 
$\gamma\cdot$profit$(\hat{R})/$profit$(R^*) = $\\
\hspace*{.2in}$\rho_1 \geq f_1(s) h_3 + \left( {1 \over 2}f_1(s)+{1 \over 2}f_2(s)\right) h_4 + f_2(s) h_5 + \left( {1 \over 2}f_2(s)+{1 \over 2}f_3(s)\right) h_6 + f_3(s) h_7$.
\end{lemma}

\begin{proof}
Windows from $H_5$ contribute $f_1(s)h_3$ in both sets of periods.  Windows from $H_6$ contribute $f_1(s)h_4$ in one set of periods and $f_2(s)h_4$ in the other.  Windows from $H_7$ contribute $f_2(s)h_5$ in both sets of periods.  Windows from $H_8$ windows contribute $f_2(s)h_6$ in one set of periods and $f_3(s)h_6$ in the other.  Finally, windows from set $H_9$ windows contribute $f_3(s)h_7$ in both sets of periods.
\end{proof}

\begin{lemma}
\label{lemma:phase2}
The second phase of SPEEDUPW12 yields a run $\hat{R}$ with $\gamma\cdot$profit$(\hat{R})/$profit$(R^*) = $\\
\hspace*{.2in}$\rho_2 \geq {1 \over 3}f_1(s)h_3 + {2 \over 3}f_1(s)h_4 + f_1(s)h_5 + \left( {2 \over 3}f_1(s)+{1 \over 3}f_2(s)
\right) h_6 + \left( {1 \over 3}f_1(s)+{2 \over 3}f_2(s)\right) h_7$.
\end{lemma}

\begin{proof}
Windows from $H_5$ contribute ${1 \over 3}f_1(s)h_3$ in one set of periods and nothing in the other two.  Windows from $H_6$ contribute ${1 \over 3}f_1(s)h_4$ in two sets of periods and nothing in the other one.  Window from $H_7$ contribute ${1 \over 3}f_1(s)h_5$ in all three sets of periods.  Windows from $H_8$ contribute ${1 \over 3}f_1(s)h_6$ in two sets of periods and ${1 \over 3}f_2(s)h_6$ in the other one.  Windows from $H_9$ contribute ${1 \over 3}f_1(s)h_7$ in one set of periods and ${1 \over 3}f_2(s)h_7$ in the other two.\end{proof}

\begin{lemma}
\label{lemma:phase3}
The third phase of SPEEDUPW12 yields a run $\hat{R}$ with $\gamma\cdot$profit$(\hat{R})/$profit$(R^*) = $\\
\hspace*{.2in}$\rho_3 \geq {1 \over 4}f_1(s)h_4 + {1 \over 2}f_1(s)h_5 + {3 \over 4}f_1(s)h_6 + f_1(s)h_7$.
\end{lemma}

\begin{proof}
Windows from $H_5$ contribute nothing in all four sets of periods.  Windows from  $H_6$ contribute ${1 \over 4}f_1(s)h_4$ in one set of periods and nothing in the other three.   Windows from $H_7$ contribute ${1 \over 4}f_1(s)h_5$ in two sets of periods and nothing in the other two.  Windows from $H_8$ contribute ${1 \over 4}f_1(s)h_6$ in three sets of periods and nothing in the other one.  Windows from $H_9$ contribute ${1 \over 4}f_1(s)h_7$ in all four sets of periods.\end{proof}

\begin{comment}
\begin{theorem}
In $O(\min \{r,m\}\Gamma(n))$ time,
SPEEDUPW12 finds a run $\hat{R}$ such that $\gamma\cdot$profit$(\hat{R})$ \\
$/$profit$(R^*) \geq$
$\max \{ ~\rho ~|~ \rho \leq f_1(s)x + {1 \over 3}f_1(s)y,$
\hspace{.1in}$\rho \leq \left( {1 \over 2}f_1(s)+{1 \over 2}f_2(s)
\right) x + {2 \over 3}f_1(s)y + {1 \over 4}f_1(s)z,$\\
\hspace*{.3in}$\rho \leq f_2(s)x + f_1(s)y + {1 \over 2}f_1(s)z,$
\hspace{.1in}$\rho \leq \left( {1 \over 2}f_2(s)+{1 \over 2}f_3(s)
\right) x + \left( {2 \over 3}f_1(s)+{1 \over 3}f_2(s)
\right) y + {3 \over 4}f_1(s)z,$\\
\hspace*{.3in}$\rho \leq f_3(s)x + \left( {1 \over 3}f_1(s)+{2 \over 3}f_2(s)
\right) y + f_1(s)z,$
\hspace{.1in}$x+y+z \leq 1,$\hspace{.1in}$x \geq 0,~y\geq 0,~z \geq 0\}$

\label{thm:LP}
\end{theorem}
\end{comment}

From Lemmas \ref{lemma:phase1}, \ref{lemma:phase2}, and \ref{lemma:phase3}, we isolate the coefficients $b_\ell$ of the variables $h_\ell$, for $\ell = 3$, 4, 5, 6, 7.  Weighting them by $x$, $y$, and $z$ to correspond to those lemmas, respectively, leads to the following definitions of five functions of $x$, $y$, $z$, and $s$.
\begin{eqnarray*}
b_3 & = & f_1(s)x + {1 \over 3}f_1(s)y\\
b_4 & = & \left( {1 \over 2}f_1(s)+{1 \over 2}f_2(s)\right) x + {2 \over 3}f_1(s)y + {1 \over 4}f_1(s)z\\
b_5 & = & f_2(s)x + f_1(s)y + {1 \over 2}f_1(s)z\\
b_6 & = & \left( {1 \over 2}f_2(s)+{1 \over 2}f_3(s)\right) x + \left( {2 \over 3}f_1(s)+{1 \over 3}f_2(s)\right) y + {3 \over 4}f_1(s)z\\
b_7 & = & f_3(s)x + \left( {1 \over 3}f_1(s)+{2 \over 3}f_2(s)\right) y + f_1(s)z
\end{eqnarray*}

\begin{theorem}
In $O(\min \{r,m\}\Gamma(n))$ time,
SPEEDUPW12 finds a run $\hat{R}$ such that $\gamma\cdot$profit$(\hat{R})$ \\
$/($profit$(R^*)) \geq \max \{~\rho~|~\rho \leq b_\ell$ for $\ell = 3,4,5,6,7$, $x + y + z \leq 1$, $x \geq 0$, $y \geq 0$, $z \geq 0$ $\}$.
\label{thm:LP}
\end{theorem}

\begin{proof}
By Lemmas \ref{lemma:phase1}, \ref{lemma:phase2}, and \ref{lemma:phase3}, {\it profit}$(\hat{R})/(${\it profit}$(R^*)\gamma ) \geq \max\{\rho_1$, $\rho_2$, $\rho_3\}$.
Then, for any convex combination of $\rho_1$, $\rho_2$, $\rho_3$, (i.e., $x,y,z \geq 0$ and $x + y + z = 1$), we have 
$${\mbox{\emph{profit}}(\hat{R}) \over \mbox{\emph{profit}}(R^*)\gamma }
% \geq \mathop{\max_{x + y + z = 1}}_{(x, y, z) \in \left(\mathbb{R}^+ \cup  \{0\}\right)^3}\left\{\min_{\ell \in \{3,4,5,6,7\}\mbox{ and }h_\ell=1}\{ \rho_1 x+ \rho_2 y + \rho_3 z\}\right\}$$
\geq \mathop{\max_{x + y + z = 1}}_{x, y, z \geq 0}\left\{\min_{\ell \in \{3,4,5,6,7\}\mbox{ and }h_\ell=1}\{ \rho_1 x+ \rho_2 y + \rho_3 z\}\right\}$$

Thus, the expression for $b_\ell$ is given by setting $h_\ell = 1$ and $h_i = 0$, where $i \neq \ell$, and summing $\rho_1$, $\rho_2$, and $\rho_3$, weighted by $x$, $y$, and $z$, respectively.  In this way, we can account for a problem instance being dominated by any set $H_\ell$ for $\ell = 3,4,5,6,7$.  No problem instance will be worse than a convex combination of all the bounds.  

Algorithm SPEEDUPW12 runs SPEEDUP a total of 12 times in the first phase, 6 times in the second phase, and 4 times in the third phase, for a total of 22 times.  We have shown that the running time of SPEEDUP is $O(\min \{r,m\}\Gamma(n))$.
\end{proof}

We give a description of $f_1(s)$, $f_2(s)$, and $f_3(s)$ in Table~\ref{table:Wratios}.
Function $f_1(s)$ comes from our work in \cite{Frederickson5}.  Function $f_2(s)$ comes from Sect.~\ref{section:performance W_2}.  Function $f_3(s)$ comes from Sect.~\ref{section:performance W_3}.

\begin{table}[!hbt]
\begin{center}
\begin{tabular}{rll}
$f_1(s) \geq$ & 
$\left\{ \begin{array}{l}
(s+1)/6 \smallskip \\ 
s/4 \smallskip \\ 
\end{array} \right.$ &
$\begin{array}{crclc}
1& \leq &s& \leq &2 \smallskip \\ 
2& \leq &s& \leq &4 \smallskip \\
\end{array}$
\\ \\
$f_2(s) \geq$ &
$\left\{ \begin{array}{l}
(s+1)/8 \smallskip \\
(2s-1)/8 \smallskip \\
1/2 \smallskip \\
(s-1)/4
\end{array} \right.$ &
$\begin{array}{crclc}
1& \leq &s&  \leq &2 \smallskip \\
2& \leq &s&  \leq &{5 \over 2} \smallskip \\
{5 \over 2}& \leq &s&  \leq &3 \smallskip \\
3& \leq &s& \leq &5
\end{array}$
\\ \\
$f_3(s) \geq$ &
$\left\{ \begin{array}{l}
(s+1)/10 \smallskip \\
s/6 \smallskip \\
(s^2 - 4s + 7)/8 \smallskip \\
(1 + 3s - s^2)/(23- 7s) \smallskip \\
1/2 \smallskip \\
(s-2)/4
\end{array} \right.$ &
$\begin{array}{crclc}
1& \leq &s&  \leq &2 \smallskip \\
2& \leq &s&  \leq &{7 \over 3} \smallskip \\
{7 \over 3}& \leq &s&  \leq &{17 \over 7} \smallskip \\
{17 \over 7}& \leq &s&  \leq &3 \smallskip \\
3& \leq &s& \leq &4 \smallskip \\
4& \leq &s& \leq &6
\end{array}$
\end{tabular}\\
\caption{Lower bounds on fractions of optimal profit collected for the sets $W_1$, $W_2$, and $W_3$, ignoring the factor of $\gamma$.}
\label{table:Wratios}
\end{center}
\end{table}

We produced the results in Table~\ref{table:speedupw12rats} by solving the linear programs of Theorem~\ref{thm:LP} for particular values of $s$
within each range, inferring the pattern for each range, and then proving the inferred pattern. 
Note that all but one of the reciprocals of the resulting ratios in terms of $s$ are nonlinear functions!

\begin{theorem}
\label{theorem:1<s<6}
For speedup $s$ in the range $1\leq s \leq 6$
and window lengths between 1 and 2,
algorithm SPEEDUPW12 produces a service run $\hat{R}$ for the repairman problem
with approximation ratio profit$(R^*)/$profit$(\hat{R})$ upper-bounded as in Table~\ref{table:speedupw12rats}.
\end{theorem}

\begin{proof}
For each possible speedup range, we show that $\gamma$ times the convex combinations of the functions given in Table \ref{table:Wratios} are never less than the reciprocals of the approximation ratios listed in Table \ref{table:speedupw12rats}.\\

\noindent When $1 \leq s \leq 2$, choose ${\displaystyle x = {50 \over 73}\mbox{, }y = {6 \over 73}\mbox{, and }z =  {17 \over 73}}$. Then, $b_3 = b_4 = \ldots = b_7 = {\displaystyle \frac{26s + 26}{219}}$.\\\bigskip

\noindent When $2 \leq s \leq 7/3$, choose ${\displaystyle x = {6s^2 + 3s \over 7s^2 + 6s + 3}\mbox{, }y = {-3s^2 + 9s \over 7s^2 + 6s + 3}\mbox{, and }z =  {4s^s - 6s + 3 \over 7s^2 + 6s + 3}}$.\medskip

\noindent Then, $b_3 = \ldots = b_7 = {\displaystyle \frac{5s^3 + 6s^2}{28s^2 + 24s + 12} }$.\\\bigskip

\noindent When $7/3 \leq s \leq 17/7$, choose ${\displaystyle x = {4s^2 + 2s \over -s^3 + 10s^2 - 3s + 2}\mbox{, }y = {3s^3 - 18s^2 + 27s \over -s^3 + 10s^2 - 3s + 2}}$, \medskip 

\noindent and ${\displaystyle z =  {4s^3 - 24s^2 + 32s - 2 \over -s^3 + 10s^2 - 3s + 2}}$.  Then, $b_3 = \ldots = b_7 = {\displaystyle \frac{s^4 - 2s^3 + 11s^2}{-4s^3 + 40s^2 - 12s + 8}}$.\\\bigskip

\noindent When $17/7 \leq s \leq 5/2$, choose ${\displaystyle x = {14s^3 - 39s^2 - 23s \over 17s^3 - 43s^2 - 35s - 23}\mbox{, }y = {-9s^3 + 54s^2 - 81s \over 17s^3 - 43s^2 - 35s - 23}}$,\medskip

\noindent and ${\displaystyle z = {12s^3 - 58s^2 + 69s - 23 \over 17s^3 - 43s^2 - 35s - 23}}$. Then, $b_3 = \ldots = b_7 = {\displaystyle \frac{11s^4 - 21s^3 - 50s^2}{68s^3 - 172s^2 - 140s - 92}}$.\\\bigskip

\noindent When $5/2 \leq s \leq 3$, choose ${\displaystyle x = {28s^3 - 120s^2 + 92s \over 73s^3 -409s^2 + 668s - 368}\mbox{, }y = {33s^3 - 189 s^2 + 264s \over 73s^3 -409s^2 + 668s - 368}}$, \medskip

\noindent and ${\displaystyle z = {12s^3 - 100s^2 + 312s - 368 \over 73s^3 -409s^2 + 668s - 368}}$.  Then, $b_3 = \ldots = b_7 = {\displaystyle \frac{39s^4 - 183s^3 + 180s^2}{292s^3 - 1636s^2 + 2672s - 1472}}$.\\\bigskip

\noindent When $3 \leq s \leq 4$, choose ${\displaystyle x = {2s^2 + 2 \over 3s^2 + 2s + 4}\mbox{, }y = {-3s^2 + 12s \over 3s^2 + 2s + 4}\mbox{, and }z =  {4s^2 - 12s + 4 \over 3s^2 + 2s + 4}}$.\medskip

\noindent Then, $b_3 = \ldots = b_7 = {\displaystyle \frac{s^3 + 6s^2}{12s^2 + 8s + 16} }$.\\\bigskip

\noindent When $4 \leq s \leq 5$, choose ${\displaystyle x = {2s - 2 \over -s + 16}\mbox{, }y = {-3s + 18 \over -s + 16}\mbox{, and }z =  0}$.\medskip

\noindent Then, $b_3 = b_7 = {\displaystyle {s + 4 \over -s + 16}}$, and $b_4 = b_6 = {\displaystyle \frac{s^2 - 6s + 45}{-4s + 64}} > {\displaystyle {s + 4 \over -s + 16}}$, whenever $s < 16$.\medskip

\noindent This follows since whenever $s < 16$, ${\displaystyle \frac{s^2 - 6s + 45}{-4s + 64} > {s + 4 \over -s + 16}}$ holds if and only if $(s - 5)^2 + 4 > 0$, which is always true.\medskip

\noindent Finally, bound $b_5 = {\displaystyle {s^2 - 8s + 37 \over -2s + 32} > {s + 4 \over -s + 16}}$ whenever $s < 16$. \medskip

\noindent This follows since whenever $s < 16$, ${\displaystyle {s^2 - 8s + 37 \over -2s + 32} > {s + 4 \over -s + 16}}$ holds if and only if $(s - 5)^2 + 20 > 0$, which is always true.\\\bigskip

\noindent When $5 \leq s \leq 6$, choose ${\displaystyle x = {8 \over -3s + 26}\mbox{, }y = {-3s + 18 \over -3s + 26}\mbox{, and }z =  0}$.\medskip

\noindent Then, $b_3 = b_7 = {\displaystyle {s - 14 \over 3s - 26}}$, and $b_4 = b_6 = {\displaystyle {2s - 20 \over 3s - 26}} \geq {\displaystyle {s - 14 \over 3s - 26}}$ whenever $s \leq 6$.\medskip

\noindent Finally, $b_5 = 1$ which is at least ${\displaystyle {s - 14 \over 3s - 26}}$  whenever $s \leq 6$.
\end{proof}

\section{Conclusion}

This paper has demonstrated the surprising versatility of the technique of trimming. Even with time windows whose lengths are not all the same, it is possible to simplify the structure of many time-constrained route-planning problems and apply an ordering that allows dynamic programming to work well.  For unrooted problems, the cost of this additional order is at most a constant reduction in the profit a run can earn.  We have extended results from our previous paper \cite{Frederickson5} so that we can characterize the way in which this reduction in profit can be offset, in part or in whole, by speedup over a hypothetical optimal benchmark when the lengths of time windows are not all uniform.  The key idea needed for this extension is to consider a diverse set of trials with a number of different period lengths for trimming and then choose the best result among all those found.  This approach makes trimming adapt to various distributions of window lengths.

We have given techniques that achieve an approximation ratio parameterized only by speedup when the ratio between the longest time window and the shortest time window is no greater than 2, but these techniques can be extended to other ranges of time window lengths.  For the general case, in which the ratio between the longest and the shortest time windows is $D$, the approximation ratio will worsen by a factor of $\log_2 D$, using an approach similar to the one we used in \cite{Frederickson6} for general length time windows without speedup.

It is worth mentioning that we have achieved approximation bounds for a few specific ranges of $s$ which are slightly better than the ones listed in Table \ref{table:speedupw12rats}.  While trying to accommodate these ranges into a coherent scheme, our analysis became so much more complex that we chose to give a more complete and readable presentation of results which are nearly as strong as the best we found.  The fact that better values are possible shows that there is potential in these techniques.

\bibliographystyle{amsplain}
\bibliography{bibliography}

\begin{thebibliography}{10}

\bibitem{Arkin}
E.~M. Arkin, J.~S.~B. Mitchell, and G.~Narasimhan.
\newblock Resource-constrained geometric network optimization.
\newblock In {\em Proc. 14th Symp. on Computational Geometry}, pages 307--316,
  New York, NY, USA, 1998. ACM.

\bibitem{Bansal}
N.~Bansal, A.~Blum, S.~Chawla, and A.~Meyerson.
\newblock Approximation algorithms for deadline-{TSP} and vehicle routing with
  time-windows.
\newblock In {\em Proc. 36th ACM Symp. on Theory of Computing}, pages 166--174,
  2004.

\bibitem{Bansal2}
N.~Bansal, H.-L. Chan, R.~Khandekar, K.~Pruhs, C.~Stein, and B.~Schieber.
\newblock Non-preemptive min-sum scheduling with resource augmentation.
\newblock In {\em Proc. 48th IEEE Symp. on Foundations of Computer Science},
  pages 614--624, Washington, DC, USA, 2007. IEEE Computer Society.

\bibitem{Bar-Yehuda}
R.~Bar-Yehuda, G.~Even, and S.~Shahar.
\newblock On approximating a geometric prize-collecting traveling salesman
  problem with time windows.
\newblock {\em J. Algorithms}, 55(1):76--92, 2005.

\bibitem{Blum3}
A.~Blum, S.~Chawla, D.~R. Karger, T.~Lane, A.~Meyerson, and M.~Minkoff.
\newblock Approximation algorithms for orienteering and discounted-reward
  {TSP}.
\newblock {\em SIAM J. Comput.}, 37(2):653--670, 2007.

\bibitem{Chekuri3}
C.~Chekuri and N.~Korula.
\newblock Approximation algorithms for orienteering with time windows.
\newblock 2007, {\tt http://arxiv.org/abs/0711.4825v1}.

\bibitem{Chekuri2}
C.~Chekuri, N.~Korula, and M.~P\'{a}l.
\newblock Improved algorithms for orienteering and related problems.
\newblock In {\em Proc. 19th ACM-SIAM Symp. on Discrete Algorithms}, pages
  661--670, Philadelphia, PA, USA, 2008. Society for Industrial and Applied
  Mathematics.

\bibitem{Chekuri}
C.~Chekuri and A.~Kumar.
\newblock Maximum coverage problem with group budget constraints and
  applications.
\newblock In {\em 7th Int. Workshop on Approximation Algorithms for
  Combinatorial Optimization Problems}, volume 3122 of {\em LNCS}, pages
  72--83. Springer, 2004.

\bibitem{Chen}
K.~Chen and S.~Har-Peled.
\newblock The orienteering problem in the plane revisited.
\newblock In {\em Proc. 22nd Symp. on Computational Geometry}, pages 247--254,
  New York, NY, USA, 2006. ACM.

\bibitem{Frederickson3}
G.~N. Frederickson and B.~Wittman.
\newblock Approximation algorithms for the traveling repairman and speeding
  deliveryman problems with unit-time windows.
\newblock In {\em APPROX-RANDOM}, volume 4627 of {\em LNCS}, pages 119--133.
  Springer, 2007.

\bibitem{Frederickson6}
G.~N. Frederickson and B.~Wittman.
\newblock Approximation algorithms for the traveling repairman and speeding
  deliveryman problems.
\newblock Journal version, in submission, available: {\tt
  http://arxiv.org/abs/0905.4444}, 2009.

\bibitem{Frederickson5}
G.~N. Frederickson and B.~Wittman.
\newblock Speedup in the traveling repairman problem with unit time windows.
\newblock In submission, available: {\tt http://arxiv.org/abs/0907.5372}, 2009.

\bibitem{Kalyanasundaram}
B.~Kalyanasundaram and K.~Pruhs.
\newblock Speed is as powerful as clairvoyance.
\newblock {\em J. ACM}, 47(4):617--643, 2000.

\bibitem{Karuno3}
Y.~Karuno, H.~Nagamochi, and T.~Ibaraki.
\newblock Better approximation ratios for the single-vehicle scheduling
  problems on line-shaped networks.
\newblock {\em Networks}, 39(4):203--209, 2002.

\bibitem{Nagarajan2}
V.~Nagarajan and R.~Ravi.
\newblock Poly-logarithmic approximation algorithms for directed vehicle
  routing problems.
\newblock In {\em APPROX-RANDOM}, volume 4627 of {\em LNCS}, pages 257--270.
  Springer, 2007.

\bibitem{Phillips}
C.~A. Phillips, C.~Stein, E.~Torng, and J.~Wein.
\newblock Optimal time-critical scheduling via resource augmentation.
\newblock {\em Algorithmica}, 32(2):163--200, 2002.

\bibitem{Tsitsiklis}
J.~N. Tsitsiklis.
\newblock Special cases of traveling salesman and repairman problems with time
  windows.
\newblock {\em Networks}, 22:263--282, 1992.

\end{thebibliography}

\appendix

\section{Coverage of Windows in Set $W_2$ when $2 \leq s \leq 3$}
\label{appendix:2<s<3 W_2}

Analysis of set $W_2$ when $2 \leq s \leq 3$ is done by considering service runs $A$, $A^R$, $A_{r - 2k}$, and $A_{r - 2k}^R$, noting that $\lambda = 2$.  The combined coverages of runs $A$, $A^R$, $A_{r - 2k}$, $A_{r - 2k}^R$, and all of their respective shifted versions are given in Table \ref{table:lowcontributions 2<s<5/2 W_2} when $k \leq r - 2k$ and in Table \ref{table:highcontributions 2<s<5/2 W_2} when $k \geq r - 2k$.  
The combined coverage of the pair $A$ and $A^R$ is exactly the same for all values of $k$ and are only listed in Table \ref{table:lowcontributions 2<s<5/2 W_2}.  These and all other tables in the Appendices are generated by using CREATE-TABLE-$\lambda$ for each different run type, using the appropriate value of $\lambda$ (2 for $W_2$ or 3 for $W_3$), and specific values for $\Delta$ determined by the number of hops each run has been moved.

\begin{table}[!hbt]
\begin{center}
\begin{tabular}{rcll}
\begin{tabular}{r}
Combined contributions\\
for $A$ and $A^R$ 
\end{tabular}
& = & 
$\left\{ \begin{array}{l}
r \smallskip \\ 
r + {1 \over 2}k - {1 \over 2}i \smallskip \\ 
{1 \over 2}r + k
\end{array} \right.$ &
$\begin{array}{rcl}
0 \leq &i& \leq k \smallskip \\ 
k \leq &i& \leq r - k \smallskip \\
r - k \leq &i& \leq \left\lfloor {3r \over 2} \right\rfloor 
\end{array}$
\\ \\
\begin{tabular}{r}
Combined contributions\\
for $A_{r - 2k}$ and $A_{r - 2k}^R$
\end{tabular}
& = &
$\left\{ \begin{array}{l}
2k + i \smallskip \\
{3 \over 2}k + {3\over 2}i \smallskip \\
r + {1 \over 2}i - {1 \over 2}k \smallskip \\
{3 \over 2}r - k
\end{array} \right.$ &
$\begin{array}{rcl}
0 \leq &i&  \leq k \smallskip \\
k \leq &i&  \leq r - 2k \smallskip \\
r - 2k \leq &i&  \leq r - k \smallskip \\
r - k \leq &i& \leq \left\lfloor {3r \over 2} \right\rfloor
\end{array}$
\end{tabular}\\
\caption{Contributions of runs for windows in $W_2$ when $2 \leq s \leq 5/2$ and $k \leq r - 2k$.}
\label{table:lowcontributions 2<s<5/2 W_2}
\end{center}
\end{table}

\begin{table}[!hbt]
\begin{center}
\begin{tabular}{rcll}
\begin{tabular}{r}
Combined contributions\\
for $A_{r - 2k}$ and $A_{r - 2k}^R$ 
\end{tabular}
& = & 
$\left\{ \begin{array}{l}
2k + i \smallskip \\
r \smallskip \\
r + {1 \over 2}i - {1 \over 2}k \smallskip \\
{3 \over 2}r - k
\end{array} \right.$ &
$\begin{array}{rcl}
0 \leq &i&  \leq r - 2k \smallskip \\
r - 2k \leq &i&  \leq k \smallskip \\
k \leq &i&  \leq r - k \smallskip \\
r - k \leq &i& \leq \left\lfloor {3r \over 2} \right\rfloor
\end{array}$
\end{tabular}\\
\caption{Contributions of runs for windows in $W_2$ when $2 \leq s \leq 5/2$ and $k \geq r - 2k$.}
\label{table:highcontributions 2<s<5/2 W_2}
\end{center}
\end{table}

\begin{lemma}\label{lemma:lowestinterval 2<s<5/2 W_2}
If the contributions from $A$ and $A^R$ are weighted by a factor of 3 and the contributions from $A_{r - 2k}$ and $A_{r - 2k}^R$ are weighted by a factor of 1, the yield for all intervals is at least $3r + 2k$.
\end{lemma}

\begin{proof}
We first consider the case when $k \leq r - 2k$, consulting Table \ref{table:lowcontributions 2<s<5/2 W_2}.

If $0 \leq i \leq k$, then the yield for $w_i$ is $3r + 2k + i$, which is at least $3r + 2k$, since $i \geq 0$.

If $k \leq i \leq r - 2k$, then the yield for $w_i$ is $3r + 3k$, which is greater than $3r + 2k$.

If $r - 2k \leq i \leq r - k$, then the yield for $w_i$ is $4r + k - i$, which is at least $3r + 2k$, since $i \leq r - k$.

If $r - k \leq i \leq \lfloor 3r/2 \rfloor$, then the yield for $w_i$ is $3r + 2k$.

We now consider the case when $k \geq r - 2k$, consulting Tables \ref{table:lowcontributions 2<s<5/2 W_2} and \ref{table:highcontributions 2<s<5/2 W_2}.  The algebra for the cases when $1 \leq i \leq r - 2k$, $k \leq i \leq r - k$, and $r - k \leq i \leq \lfloor 3r/2 \rfloor$ gives exactly the same results as the first, third, and fourth ranges from the previous part of the proof.  If $r - 2k \leq i \leq k$,  then the yield for $w_i$ is $4r \geq 3r + 2k$, since $r \geq 2k$ when $2 \leq s \leq 5/2$.
\end{proof}

\section{Coverage of Windows in Set $W_3$ when $1 \leq s \leq 2$}
\label{appendix:1<s<2 W_3}

Analysis of set $W_3$ when $1 \leq s \leq 2$ is done by considering service runs $A$, $A^R$, $A_{r - k}$, $A_{r - k}^R$, $A_{2r - k}$, $A_{2r - k}^R$, $A_{3r - k}$, and $A_{3r - k}^R$, noting that $\lambda = 3$.  

The combined coverages of runs $A$, $A^R$, $A_{r - k}$, $A_{r - k}^R$, $A_{2r - k}$, $A_{2r - k}^R$, $A_{3r - k}$, $A_{3r - k}^R$, and all of their respective shifted versions are given in Table \ref{table:lowcontributions W_3} when $k \leq r - k$ and in Table \ref{table:highcontributions W_3} when $k \geq r - k$.  The combined coverages of the pair $A$ and $A^R$ and the pair $A_{r - k}$ and $A_{r - k}^R$ are exactly the same for all values of $k$ and are only listed in Table \ref{table:lowcontributions W_3}.

\begin{table}[!hbt]
\begin{center}
\begin{tabular}{rcll}
\begin{tabular}{r}
Combined contributions\\
for $A$ and $A^R$ 
\end{tabular}
& = & 
$\left\{ \begin{array}{l}
r - {1 \over 2}i \smallskip \\ 
r + {1 \over 2}k - i \smallskip \\ 
{1 \over 2}r + {1 \over 2}k  - {1 \over 2}i \smallskip \\
0
\end{array} \right.$ &
$\begin{array}{rcl}
0 \leq &i& \leq k \smallskip \\ 
k \leq &i& \leq r \smallskip \\
r \leq &i& \leq r + k \smallskip \\
r + k \leq &i& \leq 2r
\end{array}$
\\ \\
\begin{tabular}{r}
Combined contributions\\
for $A_{r - k}$ and $A_{r - k}^R$ 
\end{tabular}
& = &
$\left\{ \begin{array}{l}
k + i \smallskip \\
{3 \over 2}r - {1\over 2}k - {1 \over 2}i \smallskip \\
2r - {1\over 2}k - i\smallskip \\
r - {1\over 2}i
\end{array} \right.$ &
$\begin{array}{rcl}
0 \leq &i& \leq r - k \smallskip \\
r - k \leq &i&  \leq r \smallskip \\
r  \leq &i& \leq 2r - k \smallskip \\
2r - k \leq &i& \leq 2r
\end{array}$
\\ \\
\begin{tabular}{r}
Combined contributions\\
for $A_{2r - k}$ and $A_{2r - k}^R$
\end{tabular}
& = & 
$\left\{ \begin{array}{l}
0 \smallskip \\ 
i - r + k \smallskip \\ 
{3 \over 2}i - {3 \over 2}r + k \smallskip \\
2i - 2r + {1 \over 2}k \smallskip \\
{1\over 2}i + r - k
\end{array} \right.$ &
$\begin{array}{rcl}
0 \leq &i& \leq r - k \smallskip \\ 
r - k \leq &i& \leq r \smallskip \\
r \leq &i& \leq r + k\smallskip \\
r + k \leq &i& \leq 2r - k \smallskip \\
2r - k \leq &i& \leq 2r
\end{array}$
\\ \\
\begin{tabular}{r}
Combined contributions\\
for $A_{3r - k}$ and $A_{3r - k}^R$
\end{tabular}
& = & 
$\left\{ \begin{array}{l}
{1 \over 2}i \smallskip \\ 
i - {1 \over 2}k \smallskip \\ 
{1 \over 2}i + {1 \over 2}r - {1 \over 2}k \smallskip \\
2r + k - i \smallskip \\
2k
\end{array} \right.$ &
$\begin{array}{rcl}
0 \leq &i& \leq k \smallskip \\ 
k \leq &i& \leq r \smallskip \\
r \leq &i& \leq r + k\smallskip \\
r + k \leq &i& \leq 2r - k \smallskip \\
2r - k \leq &i& \leq 2r
\end{array}$
\end{tabular}\\
\caption{Contributions of runs for windows in $W_3$ when $1 \leq s \leq 2$ and $k \leq r - k$.}
\label{table:lowcontributions W_3}
\end{center}
\end{table}

\begin{table}[!hbt]
\begin{center}
\begin{tabular}{rcll}
\begin{tabular}{r}
Combined contributions\\
for $A_{2r - k}$ and $A_{2r - k}^R$
\end{tabular}
& = & 
$\left\{ \begin{array}{l}
0 \smallskip \\ 
i - r + k \smallskip \\ 
{3 \over 2}i - {3 \over 2}r + k \smallskip \\
{3 \over 2}r - {1 \over 2}k \smallskip \\
{1\over 2}i + r - k
\end{array} \right.$ &
$\begin{array}{rcl}
0 \leq &i& \leq r - k \smallskip \\ 
r - k \leq &i& \leq r \smallskip \\
r \leq &i& \leq 2r - k\smallskip \\
2r - k \leq &i& \leq r + k \smallskip \\
r + k \leq &i& \leq 2r
\end{array}$
\\ \\
\begin{tabular}{r}
Combined contributions\\
for $A_{3r - k}$ and $A_{3r - k}^R$
\end{tabular}
& = & 
$\left\{ \begin{array}{l}
{1 \over 2}i \smallskip \\ 
i - {1 \over 2}k \smallskip \\ 
{1 \over 2}i + {1 \over 2}r - {1 \over 2}k \smallskip \\
{3 \over 2}i - {3\over 2}r + {1\over 2}k \smallskip \\
2k
\end{array} \right.$ &
$\begin{array}{rcl}
0 \leq &i& \leq k \smallskip \\ 
k \leq &i& \leq r \smallskip \\
r \leq &i& \leq 2r - k \smallskip \\
2r - k \leq &i& \leq r + k \smallskip \\
r + k \leq &i& \leq 2r
\end{array}$
\end{tabular}\\
\caption{Contributions of runs for windows in $W_3$ when $1 \leq s \leq 2$ and $k \geq r - k$.}
\label{table:highcontributions W_3}
\end{center}
\end{table}

\begin{lemma}\label{lemma:lowestinterval W_3}
If the contributions from $A$ and $A^R$ are weighted by a factor of 2 and the contributions from $A_{r - k}$, $A_{r - k}^R$, $A_{2r - k}$, $A_{2r - k}^R$, $A_{3r - k}$, and $A_{3r - k}^R$ are weighted by a factor of 1, the yield for all intervals is at least $2r + k$.
\end{lemma}

\begin{proof}
We first consider the case when $k \leq r - k$, consulting Table \ref{table:lowcontributions W_3}.  The algebra for the cases when $0 \leq i \leq r$ gives the same results the first, second, and third cases in Lemma \ref{lemma:lowestinterval W_2}, at least $2r + k$  in each case.  If $r \leq i \leq 2r$, then the yield for $w_i$ is $2r + k$.

We now consider the case when $k \geq r - k$, consulting Tables  \ref{table:lowcontributions W_3} and \ref{table:highcontributions W_3}.  The algebra for the cases when $0 \leq i \leq r$ gives the same results as the proof of Lemma \ref{lemma:lowestinterval W_2} for $k \geq r - k$, at least $2r + k$ in each case.  If $r \leq i \leq 2r$, then the yield for $w_i$ is again $2r + k$.
\end{proof}

\section{Coverage of Windows in Set $W_3$ when $2 \leq s \leq 7/3$}
\label{appendix:2<s<7/3 W_3}

Analysis of set $W_3$ when $2 \leq s \leq 7/3$ is done by considering service runs $A$, $A^R$, $C_{(3r - k)/2}$, and $C_{(3r - k)/2}^R$, noting that $\lambda = 3$. The combined coverages for these runs are listed in Table \ref{table:2<s<7/3 W_3} assuming that $r + k$ is even.  When $r + k$ is not even, we can achieve an identical speed by multiplying both by 2.\\

\begin{table}[!hbt]
\begin{center}
\begin{tabular}{rcll}
\begin{tabular}{r}
Combined contributions\\
for $A$ and $A^R$
\end{tabular}
& = & 
$\left\{ \begin{array}{l}
r \smallskip \\ 
r + {1 \over 2}k - {1 \over 2}i \smallskip \\ 
k
\end{array} \right.$ &
$\begin{array}{rcl}
0 \leq &i& \leq k \smallskip \\ 
k \leq &i& \leq 2r - k \smallskip \\
2r - k \leq &i& \leq 2r
\end{array}$
\\ \\
\begin{tabular}{r}
Combined contributions\\
for $C_{(3r - k)/2}$ and $C_{(3r - k)/2}^R$
\end{tabular}
& = &
$\left\{ \begin{array}{l}
{1 \over 2}r + {1 \over 2}k + {1 \over 2}i \smallskip \\
{3 \over 4}r - {1 \over 4}k \smallskip \\
{5 \over 8}r - {1 \over 8}k + {1 \over 4}i \smallskip \\
{1 \over 4}r + {1 \over 4}k + {1 \over 2}i \smallskip \\
r 
\end{array} \right.$ &
$\begin{array}{rcl}
0 \leq &i&  \leq {1 \over 2}(r - 3k) \smallskip \\
{1 \over 2}(r - 3k) \leq &i&  \leq {1 \over 2}(r - k) \smallskip \\
{1 \over 2}(r - k) \leq &i&  \leq {1 \over 2}(3r - 3k) \smallskip \\
{1 \over 2}(3r - 3k) \leq &i&  \leq {1 \over 2}(3r - k) \smallskip \\
{1 \over 2}(3r - k) \leq &i& \leq 2r
\end{array}$
\end{tabular}\\
\caption{Contributions of runs for windows in $W_3$ when $2 \leq s \leq 7/3$.}
\label{table:2<s<7/3 W_3}
\end{center}
\end{table}

\begin{lemma}\label{lemma:lowestinterval 2<s<7/3 W_3}
If the contributions from $A$ and $A^R$ are weighted by a factor of 1 and the contributions from $C_{(3r - k)/2}$ and $C_{(3r - k)/2}^R$ are weighted by a factor of 2, the yield for all intervals is at least $2r + k$.
\end{lemma}

\begin{proof}
Consulting Table \ref{table:2<s<7/3 W_3}, we first consider the case when $5k \leq r$, which implies $k \leq (r - 3k)/2$.

If $0 \leq i \leq k$, then the yield for $w_i$ is $2r + 4k + i$, which is at least $2r + k$, since $i \geq 0$.

If $k \leq i \leq (r - 3k)/2$, then the yield for $w_i$ is $2r + 3k/2 + i/2$, which is greater than $2r + k$.

If $(r - 3k)/2 \leq i \leq (r - k)/2$, then the yield for $w_i$ is $5r/2 - i/2$, which is at least $2r + k$, since $i \leq (r - k)/2$ and $r \geq 3k$.

If $(r - k)/2 \leq i \leq (3r - 3k)/2$, then the yield for $w_i$ is $9r/4 + k/4$, which is at least $2r + k$, since $r \geq 3k$.

If $(3r - 3k)/2 \leq i \leq (3r - k)/2$, then the yield for $w_i$ is $3r/2 + k + i/2$, which is at least $2r + k$, since  $i \geq (3r - 3k)/2$ and $r \geq 3k$.

If $(3r - k)/2 \leq i \leq 2r - k$, then the yield for $w_i$ is $3r + k/2 - i/2$, which is at least $2r + k$, since $i \leq 2r - k$.

If $2r - k \leq i \leq 2r$, then the yield for $w_i$ is $2r + k$.

Next, we consider the case when $5k \geq r$, again consulting Table \ref{table:2<s<7/3 W_3}.

The first case gives the same result as above but for the range $0 \leq i \leq (r - 3k)/2$.  

If $(r - 3k)/2 \leq i \leq k$, then the yield for $w_i$ is $5r/2 - k/2$, which is at least $2r + k$, since $r \geq 3k$.

The third case gives the same result as above but for the range $k \leq i \leq (r - k)/2$.  The fourth, fifth, and sixth cases above give identical results when $5k \geq r$.
\end{proof}

\section{Coverage of Windows in Set $W_3$ when $7/3 \leq s \leq 17/7$}
\label{appendix:7/3<s<17/7 W_3}

Analysis of set $W_3$ when $7/3 \leq s \leq 17/7$ is done by considering service runs $A$, $A^R$, $C_{2r - 2k}$, and $C_{2r - 2k}^R$, noting that $\lambda = 3$. The combined coverages for runs $A$ and $A^R$ are listed in Table \ref{table:2<s<7/3 W_3}, and the combined coverages for runs $C_{2r - 2k}$ and $C_{2r - 2k}^R$ are listed in Table \ref{table:7/3<s<17/7 W_3}, assuming that $r + k$ is even.\\

\begin{table}[!hbt]
\begin{center}
\begin{tabular}{rcll}
\begin{tabular}{r}
Combined contributions\\
for $C_{2r - 2k}$ and $C_{2r - 2k}^R$
\end{tabular}
& = &
$\left\{ \begin{array}{l}
{3 \over 4}r - {1 \over 4}k \smallskip \\
{1 \over 2}r + {1 \over 4}k + {1 \over 4}i \smallskip \\
{1 \over 4}r + {1 \over 4}k + {1 \over 2}i \smallskip \\
{3 \over 4}r + {3 \over 4}k \smallskip \\
{1 \over 4}r + k + {1 \over 4}i \smallskip \\
{1 \over 2}r + {3 \over 2}k 
\end{array} \right.$ &
$\begin{array}{rcl}
0 \leq &i&  \leq r - 2k \smallskip \\
r - 2k \leq &i&  \leq r \smallskip \\
r \leq &i&  \leq r + k \smallskip \\
r + k \leq &i&  \leq 2r - k \smallskip \\
2r - k \leq &i&  \leq r + 2k \smallskip \\
r + 2k \leq &i& \leq 2r
\end{array}$
\end{tabular}\\
\caption{Contributions of runs for windows in $W_3$ when $7/3 \leq s \leq 17/7$.}
\label{table:7/3<s<17/7 W_3}
\end{center}
\end{table}

\begin{lemma}\label{lemma:lowestinterval 7/3<s<17/7 W_3}
If the contributions from $A$ and $A^R$ are weighted by a factor of $k$ and the contributions from $C_{2r - 2k}$ and $C_{2r - 2k}^R$ are weighted by a factor of $r - k$, the yield for all intervals is at least $3r^2/4 + k^2/4$.
\end{lemma}

\begin{proof}
Consult Tables \ref{table:2<s<7/3 W_3} and \ref{table:7/3<s<17/7 W_3}.

If $0 \leq i \leq r - 2k$, then the yield for $w_i$ is $3r^2/4 + k^2/4$.

If $r - 2k \leq i \leq k$, then the yield for $w_i$ is $r^2/2 + 3rk/4 + (r - k)i/4$, which is greater than $3r^2/4 + k^2/4$, since $i \geq r - 2k$.

If $k \leq i \leq r$, then the yield for $w_i$ is $r^2/2 + 3rk/4 + k^2/4 + (r - 3k)i/4$, which is at least $3r^2/4 + k^2/4$, since $i \leq r$ and $r \leq 3k$.

If $r \leq i \leq r + k$, then the yield for $w_i$ is $r^2/4 + rk + k^2/4 + (r - 2k)i/2$, which is at least $3r^2/4 + k^2/4$, since $i \geq r$.

If $r + k \leq i \leq 2r - k$, then the yield for $w_i$ is $3r^2/4 + rk - k^2/4 - ki/2$, which is at least $3r^2/4 + k^2/4$, since  $i \leq 2r - k$.

If $2r - k \leq i \leq r + 2k$, then the yield for $w_i$ is $r^2/4 + 3rk/4 + (r - k)i/4$, which is at least $3r^2/4 + k^2/4$, since $i \geq 2r - k$.

If $r + 2k \leq i \leq 2r$, then the yield for $w_i$ is $r^2/2 + rk - k^2/2$, which is at least $3r^2/4 + k^2/4$, since $(r^2/2 + rk - k^2/2) - (3r^2/4 + k^2/4) = ((r - k)/2)((3k - r)/2) \geq 0$ when $k \leq r \leq 3k$.
\end{proof}

\section{Coverage of Windows in Set $W_3$ when $17/7 \leq s \leq 3$}
\label{appendix:17/7<s<3 W_3}

Analysis of set $W_3$ when $17/7 \leq s \leq 3$ is done by considering service runs $A$, $A^R$, $B_{r - k + 1}$, and $B_{r - k + 1}^R$, noting that $\lambda = 3$.  The combined coverages for runs $A$ and $A^R$ are listed in Table \ref{table:2<s<7/3 W_3}, and the combined coverages for runs $B_{r - k + 1}$ and $B_{r - k + 1}^R$ are listed in Table \ref{table:17/7<s<5/2 W_3} when $17/7 \leq s \leq 5/2$ and in Table \ref{table:5/2<s<3 W_3} when $5/2 \leq s \leq 3$, assuming in both cases that $r + k$ is even.\\

\begin{table}[!hbt]
\begin{center}
\begin{tabular}{rcll}
\begin{tabular}{r}
Combined contributions\\
for $B_{r - k + 1}$ and $B_{r - k + 1}^R$
\end{tabular} 
& = &
$\left\{ \begin{array}{l}
{2 \over 3}k + {2 \over 3}i \smallskip \\
-{1 \over 3}r + k + i\smallskip \\
-{1 \over 2}r + {5 \over 6}k + {4 \over 3}i \smallskip \\
{1 \over 3}k + i \smallskip \\
{2 \over 3}r + {2 \over 3}i \smallskip \\
{5 \over 3}r + {1 \over 3}k 
\end{array} \right.$ &
$\begin{array}{rcl}
0 \leq &i&  \leq r - k \smallskip \\
r - k \leq &i&  \leq {1 \over 2}(r + k) \smallskip \\
{1 \over 2}(r + k) \leq &i&  \leq {1 \over 2}(3r - 3k) \smallskip \\
{1 \over 2}(3r - 3k) \leq &i&  \leq 2r - k \smallskip \\
2r - k \leq &i&  \leq {1 \over 2}(3r + k) \smallskip \\
{1 \over 2}(3r + k) \leq &i& \leq 2r
\end{array}$
\end{tabular}\\
\caption{Contributions of runs for windows in $W_3$ when $17/7 \leq s \leq 5/2$.}
\label{table:17/7<s<5/2 W_3}
\end{center}
\end{table}

\begin{table}[!hbt]
\begin{center}
\begin{tabular}{rcll}
\begin{tabular}{r}
Combined contributions\\
for $B_{r - k + 1}$ and $B_{r - k + 1}^R$
\end{tabular}
& = &
$\left\{ \begin{array}{l}
{2 \over 3}k + {2 \over 3}i \smallskip \\
-{1 \over 3}r + k + i\smallskip \\
{1 \over 6}r + {1 \over 2}k + {2 \over 3}i \smallskip \\
{1 \over 3}k + i \smallskip \\
{2 \over 3}r + {2 \over 3}i \smallskip \\
{5 \over 3}r + {1 \over 3}k 
\end{array} \right.$ &
$\begin{array}{rcl}
0 \leq &i&  \leq r - k \smallskip \\
r - k \leq &i&  \leq {1 \over 2}(3r - 3k) \smallskip \\
{1 \over 2}(3r - 3k) \leq &i&  \leq {1 \over 2}(r + k) \smallskip \\
{1 \over 2}(r + k) \leq &i&  \leq 2r - k \smallskip \\
2r - k \leq &i&  \leq {1 \over 2}(3r + k) \smallskip \\
{1 \over 2}(3r + k) \leq &i& \leq 2r
\end{array}$
\end{tabular}\\
\caption{Contributions of runs for windows in $W_3$ when $5/2 \leq s \leq 3$.}
\label{table:5/2<s<3 W_3}
\end{center}
\end{table}

\begin{lemma}\label{lemma:lowestinterval 17/7<s<3 W_3}
If the contributions from $A$ and $A^R$ are weighted by a factor of $6r - 4k$ and the contributions from $B_{r - k + 1}$ and $B_{r - k + 1}^R$ are weighted by a factor of $3r - 3k$, the yield for all intervals is at least $6r^2 - 2rk - 2k^2$.
\end{lemma}

\begin{proof}
In the case that $17/7 \leq s \leq 5/2$, consult Tables \ref{table:2<s<7/3 W_3} and \ref{table:17/7<s<5/2 W_3}.

If $0 \leq i \leq k$, then the yield for $w_i$ is $6r^2 - 2rk - 2k^2 + 2(r - k)i$, which is at least $6r^2 - 2rk - 2k^2$, since $r \geq k$.

If $k \leq i \leq r - k$, then the yield for $w_i$ is $6r^2 + rk - 4k^2 - ir = 6r^2 -2rk -2k^2 +(r-k-i)r +(r-2k)k +(3k-r)r$, which is greater than $6r^2 - 2rk - 2k^2$, since $i \leq r - k$, $r \geq 2k$, and $k > r/3$.

If $r - k \leq i \leq (r + k)/2$, then the yield for $w_i$ is $5r^2 + 3rk - 5k^2 - ik = 6r^2 -2rk -2k^2 +((r+k)/2 -i)k +(r-2k)7k/4 +(7k-3r)r/3 +5rk/12$, which is greater than $6r^2 - 2rk - 2k^2$, since $i \leq (r + k)/2$, $r \geq 2k$, and $k \geq 3r/7$.

If $(r + k)/2 \leq i \leq (3r - 3k)/2$, then the yield for $w_i$ is $9r^2/2 + 3rk - 9k^2/2 + (r - 2k)i = 6r^2 -2rk -2k^2 +(i-(r+k)/2)(r-2k) +(r-2k)7k/4 +(7k-3r)r/3 +5rk/12$, which is greater than $6r^2 - 2rk - 2k^2$, since $i \geq (r + k)/2$, $r \geq 2k$, and $k\geq 3r/7$.

If $(3r - 3k)/2 \leq i \leq 2r - k$, then the yield for $w_i$ is $6r^2 - 3k^2 - ik$, which is at least $6r^2 - 2rk - 2k^2$, since  $i \leq 2r - k$.

If $2r - k \leq i \leq (3r + k)/2$, then the yield for $w_i$ is $2r^2 + 4rk - 4k^2 + 2(r - k)i$, which is at least $6r^2 - 2rk - 2k^2$, since $i \geq 2r - k$.

If $(3r + k)/2 \leq i \leq 2r$, then the yield for $w_i$ is $5r^2 + 2rk - 5k^2 = 6r^2 -2rk -2k^2 +(r-2k)3k/2 +(7k-3r)r/3 + rk/6$, which is at least $6r^2 - 2rk - 2k^2$, since $r \geq 2k$ and $k \geq 3r/7$.

In the case that $5/2 \leq s \leq 3$, consult Tables \ref{table:2<s<7/3 W_3} and \ref{table:5/2<s<3 W_3}.  For this range $(3r - 3k)/2 \leq (r + k)/2$.

If $r - k \leq i \leq (3r - 3k)/2$, then the yield for $w_i$ is $5r^2 + 3rk - 5k^2 - ik = 6r^2 -2rk -2k^2 +((3r-3k)/2 -i)k +(r-k)3k/2 +(2k-r)r$, which is at least $6r^2 - 2rk - 2k^2$, since $i \leq (3r - 3k)/2$ and $k \leq r \leq 2k$.

If $(3r - 3k)/2 \leq i \leq (r + k)/2$, then the yield for $w_i$ is $13r^2/2 - 7k^2/2 - ir$, which is at least $6r^2 - 2rk - 2k^2$, since $i \leq (r + k)/2$ and $r \geq k$.

If $(r + k)/2 \leq i \leq 2r - k$, then the yield for $w_i$ is $6r^2 - 3k^2 - ik$, which is at least $6r^2 - 2rk - 2k^2$, since $i \leq 2r - k$.

All other ranges are identical to some yield when $17/7 \leq s \leq 5/2$.
\end{proof}

\end{document}